\documentclass[11pt,a4paper,oneside]{article}
\usepackage[top=3cm, bottom=3cm, left=2cm, right=2cm]{geometry}
\usepackage[english]{babel}
\usepackage[utf8x]{inputenc}
\usepackage[T1]{fontenc}

\usepackage{amsthm,amsmath,amssymb,mathrsfs,dsfont,mathtools}

\usepackage{bm}
\usepackage{graphicx}
\usepackage[colorinlistoftodos]{todonotes}
\usepackage[colorlinks=true, allcolors=blue]{hyperref}
\usepackage{booktabs,subcaption}
\usepackage{hyperref}
\usepackage[all]{nowidow}
\usepackage{setspace}
\usepackage{algorithm}
\usepackage[noend]{algpseudocode}
\usepackage{tikz} 
\usepackage[sf]{titlesec}
\usepackage{sectsty}
\allsectionsfont{\centering \normalfont\scshape} 
\usepackage{lipsum}
\usepackage{adjustbox}
\usepackage{lineno} 
\usepackage{enumitem} 
\usepackage{natbib}

\usepackage{authblk}
\title{Bayesian nonparametric modeling of latent partitions via Stirling-gamma priors}
\date{}
\author[1]{Alessandro Zito}
\author[2]{Tommaso Rigon} 
\author[3]{David B. Dunson}

\affil[1]{Department of Biostatistics, Harvard T.H. Chan School of Public Health, Boston, MA, 02115, U.S.A.}
\affil[3]{Department of Statistical Science, Duke University, Durham, NC, 27708, U.S.A.}
\affil[2]{Department of Economics, Management and Statistics, University of Milano-Bicocca, Milan, 20126, Italy}

\newtheorem{theorem}{Theorem}
\newtheorem{corollary}{Corollary}
\newtheorem{lemma}{Lemma}
\newtheorem{proposition}{Proposition}

\theoremstyle{definition}
\newtheorem{definition}{Definition}
\newtheorem{remark}{Remark}

\newcommand{\Vabm}{\mathcal{S}_{a, b , m}}

\definecolor{revision_color}{HTML}{00008B}

\usepackage[resetlabels]{multibib}
\newcites{Supp}{References}

\begin{document}

\maketitle

\begin{abstract}
Dirichlet process mixtures are particularly sensitive to the value of the precision parameter controlling the behavior of the latent partition. Randomization of the precision through a prior distribution is a common solution, which leads to more robust inferential procedures. However, existing prior choices do not allow for transparent elicitation, due to the lack of analytical results. We introduce and investigate a novel prior for the Dirichlet process precision, the Stirling-gamma distribution.  We study the distributional properties of the induced random partition, with an emphasis on the number of clusters. Our theoretical investigation clarifies the reasons of the improved robustness properties of the proposed prior. Moreover, we show that, under specific choices of its hyperparameters, the Stirling-gamma distribution is conjugate to the random partition of a Dirichlet process. We illustrate with an ecological application the usefulness of our approach for the detection of communities of ant workers. 
\end{abstract}

\section{Introduction}

Discrete Bayesian nonparametric priors have been thoroughly investigated in recent decades motivated by their wide applicability in model-based clustering and density estimation problems. Suppose $X_1, \ldots, X_n$ are $n$ observations taking values on $\mathds{X}$ and  $f(x\mid \theta)$ is a density function on the same space, indexed by $\theta$. Then, a Bayesian nonparametric mixture model is defined through the following hierarchical representation:
\begin{equation}\label{eq:MixtModel}
X_i\mid \theta_i \stackrel{\text{ind}}{\sim} f(x\mid \theta_i), \qquad 
\theta_i \mid \tilde{p}  \stackrel{\text{iid}}{\sim} \tilde{p}, \qquad \tilde{p}  \sim \mathscr{Q}, \qquad (i=1,\dots,n),
\end{equation}
where $\theta_1,\dots,\theta_n$ are latent parameters, $\tilde{p}$ is a discrete random probability measure and $\mathscr{Q}$ represents its prior. Some notable instances of prior laws $\mathscr{Q}$ include the Pitman--Yor process \citep{Perman_1992, Pitman1997}, Gibbs-type priors \citep{Gnedin2005, DeBlasi2015}, mixtures of finite mixtures and their generalizations \citep{Richardson_green_1997, Miller_Harrison_2018, FruFru_2021}, and normalized random measures with independent increments \citep{Regazzini_2003}. Arguably, the most popular and widely employed discrete nonparametric prior is the Dirichlet process introduced by \citet{Ferguson1973},  due to its simplicity and analytical tractability. 

The discreteness of $\tilde{p}$ induces a clustering of the observations by generating ties among the latent parameters. More precisely, there will be $K_n=k$ distinct values among $\theta_1, \ldots, \theta_n$, which partitions the statistical units $\{1, \ldots, n\}$ into $k$ clusters, say $C_1,\ldots, C_k$.  Hence, two statistical units $i$ and $i'$ belong to the same cluster, say the $j$th, if $i,i' \in C_j$ or, equivalently, if $\theta_i = \theta_{i'}$. Moreover, we will say that $\Pi_n = \{C_1,\ldots, C_k\}$ is the random partition induced by $\tilde{p}$. In a Dirichlet process mixture model, the law of such a random partition $\Pi_n$ is
\begin{equation}\label{eq:dp_eppf}
\mathds{P}(\Pi_n = \{C_1, \ldots, C_k\}\mid \alpha) = \frac{\alpha^k}{(\alpha)_n}\prod_{j=1}^k (n_j - 1)!,
\end{equation}
where $\alpha > 0$, with $(\alpha)_n = \alpha(\alpha+1)\cdots(\alpha + n - 1)$ being the ascending factorial, with $(\alpha)_0 = 1$ and with $n _j = |C_j|$ being the number of elements in cluster $C_j$, so that $\sum_{j=1}^k n_j = n$. The parameter $\alpha$ is called the \emph{precision} and, together with the sample size $n$, governs the law of the partition and the number of clusters $K_n$. In our motivating application, we rely on such a random partition mechanism to infer the latent communities in a colony of ant workers. Specifically, we model individual ant-to-ant interaction networks via stochastic block models \citep{Nowicki_Snijders_2001}, which are a variant of the mixture model in~\eqref{eq:MixtModel}. See \citet{Kemp_2006,Geng_et_al_2019, Legramanti2022} for other applications of discrete nonparametric priors in community detection tasks. 

\begin{figure}[tb]
\includegraphics[width = \linewidth]{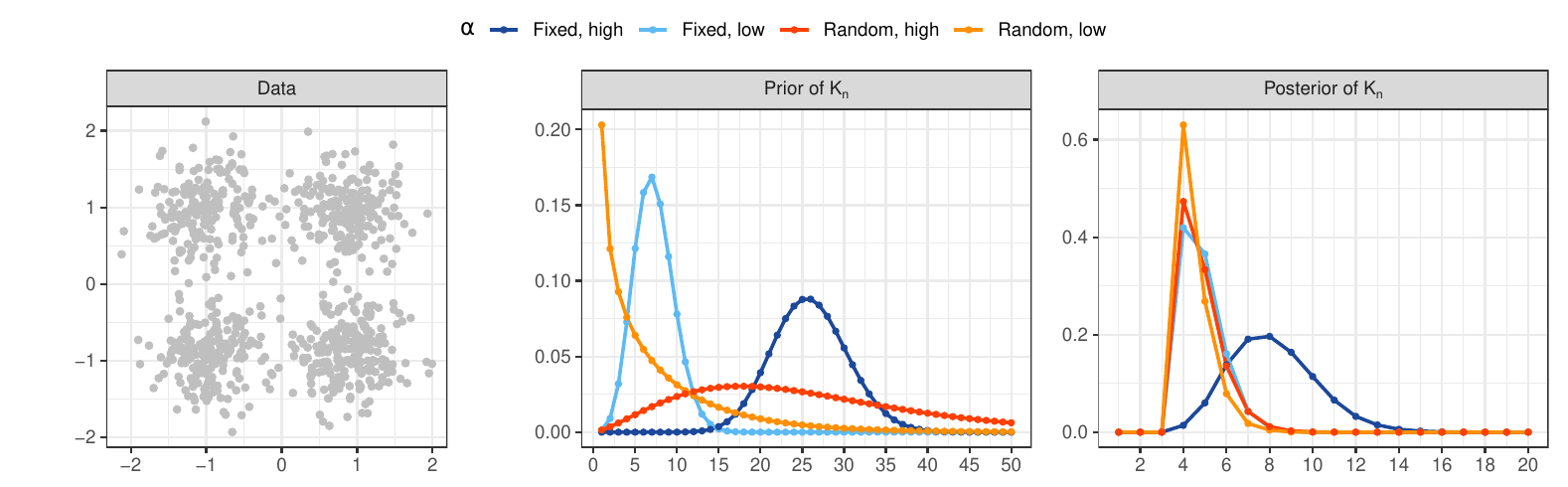}
    \caption{Left panel: $800$ data points from a four-component mixture of normals. Center panel: probability mass function of the prior distribution of $K_n$ under different choices of $\alpha$.  Parameters were set to have $\mathds{E}(K_n) = 7.26$ and $\mathds{E}(K_n) = 25.9$ in low and high cases, respectively, with $\alpha=1$ and $\alpha = 5$ in the fixed cases. Right panel: posterior distribution of $K_n$ estimated from the data via a Dirichlet process mixture. See the Supplementary material for details.}
    \label{fig:figure1}
\end{figure}

It has been pointed out by several scholars that Dirichlet process mixtures are particularly sensitive to the precision parameter \citep{Escobar_1994, lijoi2007_jrrsb, Booth_el_al_2008}. For instance, different values of $\alpha$ can lead to dramatically different posterior distributions of $K_n$, even when sufficient cluster separation is present in the data. 
Such a lack of robustness is problematic when the posterior partition is of inferential interest, such as in clustering and community detection.  However, as we illustrate in Figure~\ref{fig:figure1} with a simulated example, randomizing the precision through the use of a prior $\pi(\alpha)$ can attenuate this unpleasant behavior. Here, fixing $\alpha = 1$ as opposed to $\alpha = 5$ causes the posterior mode of $K_n$ to shift from four to eight clusters, even if the data are generated from a mixture with four well-separated components. On the contrary, allowing $\alpha$ to be random induces more flexibility in the prior for $K_n$, and in turn, yields two posterior distributions that are similar to each other even when the means in the priors for $\alpha$ are far apart. A similar behavior persists asymptotically: in a notorious example, \citet{Miller_2013, JMLR:v15:miller14a} showed that the posterior of $K_n$ from Dirichlet process mixtures with \emph{fixed} $\alpha$ does not concentrate at the true number of clusters $k^*$ when data are generated from a finite mixture of exactly $k^*$ components. Nevertheless, this \emph{inconsistency} can be prevented in certain scenarios through appropriate choices of $\pi(\alpha)$, as recently demonstrated by \citet{Ascolani_2022}.

The most commonly adopted prior $\pi(\alpha)$ is the gamma distribution, whose use was first popularized by \citet{Escobar_West_1995}. Another valid alternative is the Jeffrey's prior detailed in  \citet{RODRIGUEZ20131539}. See also \citet{DORAZIO20093384, MURUGIAH20121947} for additional examples.
However, such distributions lead to an analytically intractable prior over the partition.
This prevents transparent elicitation of the prior hyperparameters and complicates the inclusion of available  information on the clustering structure of the data. Moreover, while it has been shown that the distribution of $K_n$ arising from a Dirichlet process can be approximated with a Poisson distribution when $\alpha$ is fixed, no such approximation is 
available for the random $\alpha$ case. We aim at filling this gap by introducing a novel prior over the Dirichlet process precision that (i) is simple and easily sampled from, (ii) makes the induced prior on the random partition analytically tractable and (iii) leads to an approximate negative binomial prior on the number of clusters. Our proposed prior for $\alpha$ has a novel distribution, which we refer to as \emph{Stirling-gamma}, due to its connection with Stirling numbers and the gamma distribution. Under an appropriate logarithmic rescaling, the Stirling-gamma is equivalent to the gamma in a limiting case.

When $\alpha$ follows a Stirling-gamma prior, we will say that the random partition is from a \emph{Stirling-gamma process}. This belongs to the larger class of Gibbs-type partition models, which are discrete nonparametric priors that enjoy several appealing theoretical properties. See for instance \citet{lijoi2007_jrrsb, Lijoi_2007_a, Lijoi_prunster_walker_2008, Lijoi_prunster_walker_sinica, Favaro_lijoi_prunster_2013} and refer to  \citet{DeBlasi2015} for an in-depth overview. We provide several distributional results for the Stirling-gamma process. In particular, we show that the hyperparameters have an interpretable link with the induced law for the partition and the associated number of clusters. The resulting negative binomial-type behavior of the Stirling-gamma process, as opposed to the Poisson-type one of the Dirichlet process, helps explain the greater robustness of mixture models with random $\alpha$. 

The Stirling-gamma has the further fundamental advantage of being the conjugate prior to the law of the random partition of the Dirichlet process if one of its hyperparameters equals $n$. This happens because the distribution in equation~\eqref{eq:dp_eppf} belongs to the class of  natural exponential families, which always admit a conjugate prior \citep{Diaconis_Yilvisaker_1979}.  We illustrate how this conjugacy result further facilitates both posterior inferences on $\alpha$ and prior elicitation. To broaden the use of our proposed prior in applications, we present an efficient sampler that can be utilized when drawing from the full conditional for $\alpha$ in Dirichlet process mixtures. In this respect, an in-depth comparison with the gamma distribution is provided in simulated settings. The consequences of the prior dependency on $n$ are thoroughly discussed. In particular, we show  how the Stirling-gamma can be a useful prior in a simulated case where Dirichlet process mixtures are inconsistent for the true number of clusters \citep{Miller_2013}, and in applied settings when modeling independently repeated partitions of the same $n$ statistical units, such as the ant worker interaction networks of our illustrative application. 

The paper is structured as follows. Section~\ref{sec:dist_theory} presents the Stirling-gamma distribution, the Stirling-gamma process and its properties in relation to the general class of Gibbs-type partitions. Section~\ref{sec:inference} discusses the conjugate Stirling-gamma prior and subsequent inferential implications. Section~\ref{sec:computations} presents a random sample generator for the proposed distribution, which is also made publicly available in the \texttt{R} package \texttt{ConjugateDP}. Sections~\ref{sec:simulations} and~\ref{sec:ants} illustrate the usefulness of the Stirling-gamma prior in simulated and applied settings. Concluding remarks and  new directions are discussed in Section~\ref{sec:discussion}. 

\section{Distribution theory for Stirling-gamma processes}\label{sec:dist_theory}
\subsection{Background}\label{subsec:Gibbs}
Before introducing the Stirling-gamma distribution and the related process, we provide a probabilistic background on partition models that will be useful throughout the paper. Suppose that the latent parameters $\theta_i$ in model~\eqref{eq:MixtModel} belong to an infinite exchangeable sequence $(\theta_n)_{n \ge 1}$ and that they live in a complete and separable metric space $\Theta$ endowed with a Borel $\sigma$-algebra $\mathscr{B}(\Theta)$. The \emph{species sampling models} introduced by \citet{Pitman1996} provide a broad class of discrete nonparametric priors. More precisely, a \emph{proper} species sampling model is defined as 
$\tilde{p} = \sum_{j=1}^\infty \tilde{p}_j\delta_{\xi_j}$ with $\sum_{j=1}^\infty \tilde{p}_j = 1 \textrm{ a.s.}$, where $\delta_x$ is the Dirac measure at $x$, while the $\xi_j$s are drawn independently from a non-atomic \emph{baseline distribution} $P_0$ on $\mathscr{B}(\Theta)$ and are also independent from the random weights $\tilde{p}_j \geq 0$. Since the realizations of a proper species sampling model are almost surely discrete, we have  $\mathds{P}(\theta_i = \theta_{i'}) >0$ for any  $i\neq i'$. As such, the latent variables  $\theta_1,\ldots, \theta_n$ will take on $K_n = k$ distinct values, called $\theta_1^*, \ldots, \theta_k^*$, with frequencies $n_1,\ldots, n_k$ and $\sum_{j=1}^k n_j = n$. This induces a random partition of the statistical units $\{1,\ldots, n\}$ into groups $C_1, \ldots, C_k$, where $C_j = \{i: \theta_i = \theta_j^*\}$ for $j =1, \ldots, k$.

There exists a rich variety of exchangeable priors to model the random partition mechanism generating the clusters $C_1, \ldots, C_k$. See \citet{ghosal_van_der_vaart_2017} for an extensive account. Among them, \emph{Gibbs-type} processes \citep{Gnedin2005, DeBlasi2015} form a 
particularly rich class. We say that the law of $\tilde{p}$ is of Gibbs-type if 
\vspace{-.05em}
\begin{equation}\label{eq:eppf}
\mathds{P}(\Pi_n = \{C_1,\ldots, C_k\}) = V_{n, k} \prod_{j=1}^k (1 -\sigma)_{n_j - 1}, 
\end{equation}
\vspace{-.05em}
where $\sigma < 1$ and the coefficients $V_{n, k}$ satisfy the forward recursion
$V_{n, k} = (n - \sigma k) V_{n + 1, k} + V_{n+1, k +1}$ for all $k = 1, \ldots, n$ and $n\geq 1$, with $V_{1,1} = 1$. Equation~\eqref{eq:eppf} is the so-called \emph{exchangeable partition probability function} of the process \citep{Pitman1996}. This depends on the cluster frequencies through a product structure, which implies that Gibbs-type priors are a special instance of product partition models \citep{Hartigan1990, Barry_Hartigan_1992, Quintana_iglesias_2003}. The coefficients $V_{n, k}$ determine the system of predictive equations of $\Pi_n$, that is\vspace{-.05em}
\begin{equation}\label{eq:Gibbs_scheme}
\mathds{P}(\theta_{n+1} \in A\mid \theta_{1}, \ldots, \theta_{n}) = \frac{V_{n+1, k+1}}{V_{n, k}}P_0(A) +  \frac{V_{n+1, k}}{V_{n, k}}\sum_{j = 1}^k (n_j - \sigma) \delta_{\theta^*_j}(A),
\end{equation}
for $n\geq 1$ and every $A \in \mathscr{B}(\Theta)$. The $(n+1)$st latent parameter $\theta_{n+1}$ is drawn from the baseline $P_0$ with probability $V_{n+1,k+1}/V_{n, k}$, and is equal to one of the previous $\theta_{j}^*$ with probability $V_{n+1,k}(n_j - \sigma)/V_{n, k}$. Specifically, sampling $\theta_{n+1}$ from the baseline automatically generates a new cluster due to the diffuse nature of $P_0$. Refer to \citet{DeBlasi2015} for an overview.

When $\sigma =0 $ and $V_{n, k} = \alpha^k/(\alpha)_n$ in equation~\eqref{eq:eppf}, one recovers the exchangeable partition probability function of a Dirichlet process in equation~\eqref{eq:dp_eppf}.  A more robust specification can be obtained by introducing a prior for $\alpha$. In this case, the resulting distribution is
\begin{equation}\label{eq:eppf_mixtDp}
\mathds{P}(\Pi_n = \{C_1, \ldots, C_k\}) = V_{n,k}\prod_{j = 1}^k (n_j - 1)!, \quad V_{n, k} = \int_{\mathds{R}_+}\frac{\alpha^k}{(\alpha)_n}\pi(\alpha) \text{d}\alpha, 
\end{equation}
which has more flexibility through varying the hyperparameters of $\pi(\alpha)$. \citet{Gnedin2005} show that every Gibbs-type prior with $\sigma = 0$ is uniquely characterized by equation~\eqref{eq:eppf_mixtDp}. Commonly adopted priors $\pi(\alpha)$, such as the gamma distribution proposed by \citet{Escobar_West_1995}, do not lead to an analytically tractable form for $V_{n, k}$. This is a crucial point because $V_{n, k}$ are the key quantities that determine the distribution of the number of clusters, namely
\begin{equation}\label{eq:Kn_dist}
    \mathds{P}(K_n = k) = V_{n,k}|s(n, k)|, \qquad (k =1, \ldots, n),
\end{equation}
where $|s(n, k)|$ are the signless Stirling number of the first kind \citep{Charalambides_2005}. Refer to \citet{antoniak1974} and \citet{Gnedin2005} for derivations. Thus, our goal is to develop a prior whose hyperparameters have a clear and interpretable link with the distribution of $K_n$ in equation~\eqref{eq:Kn_dist}. In what follows, we show how this can be achieved using a Stirling-gamma prior.

\vspace{-.1em}
\subsection{The Stirling-gamma distribution}\label{subsec:SG}
In this Section, we introduce the Stirling-gamma distribution and describe its properties.
\begin{definition}\label{def:StirlingGammaPdf}
A positive random variable follows a Stirling-gamma distribution with parameters $a, b > 0$ and $m \in \mathds{N}$ satisfying $1 < a/b < m$, if its density function is
$$
\pi(\alpha) = \frac{1}{\mathcal{S}_{a,b,m}} \frac{\alpha^{a-1}}{\{(\alpha)_m\}^b}, \qquad \mathcal{S}_{a,b,m} = \int_{\mathds{R}_+} \frac{\alpha^{a-1}}{\{(\alpha)_m\}^{b}}\,\mathrm{d}{\alpha}.
$$
We will write  $\alpha \sim \mathrm{Sg}(a,b,m)$.
\end{definition}
The name of the Stirling-gamma distribution stems from the presence of the ascending factorial in the density function, whose polynomial expansion defines Stirling numbers of the first kind \citep{Charalambides_2005}, and the close connection with the gamma distribution. Indeed, from $(\alpha)_n = \alpha(\alpha + 1) \cdots (\alpha + n-1)$, the above density is easily seen as a special case in the \emph{generalized gamma convolution} class of distributions \citep[][equation 5.1.2]{Bondesson1992}, which are defined as continuous scale mixtures of gammas.  Moreover, the following result holds. 
\begin{proposition}\label{pro:limit_gamma}
Let $\alpha \sim \mathrm{Sg}(a, b, m)$. Then, the following convergence in distribution holds: $$\alpha \log{m} \to \gamma, \quad \gamma \sim \mathrm{Ga}(a-b, b), \quad m\to \infty.$$
\end{proposition}

In the above statement, $\mathrm{Ga}(a_0, b_0)$ denotes the gamma distribution with mean $a_0/b_0$ and variance $a_0/b_0^2$. Proposition~\ref{pro:limit_gamma} has two fundamental implications. The first is that the density of the Stirling-gamma $\mathrm{Sg}(a,b, m)$ progressively resembles that of $\mathrm{Ga}(a - b, b\log{m})$ as $m$ becomes larger. The second is that $\alpha\to0$ in probability as $m\to\infty$ with a logarithmic rate of convergence via a direct application of Slutzky's theorem. Both properties are illustrated in Figure~\ref{fig:Stirling_and_gamma}, which displays the probability density function of the two distributions for varying values of $m$ and $b$ when $a = 5$. In particular, high values for $a/b$ require a larger $m$ to make the two densities indistinguishable. Both distributions progressively shift towards zero as $m$ increases. However, as we show formally in the Supplementary material, the Stirling-gamma is a heavy-tailed distribution; hence, it has a heavier right tail than the gamma distribution.
\begin{figure}[t]
\includegraphics[width = \linewidth]{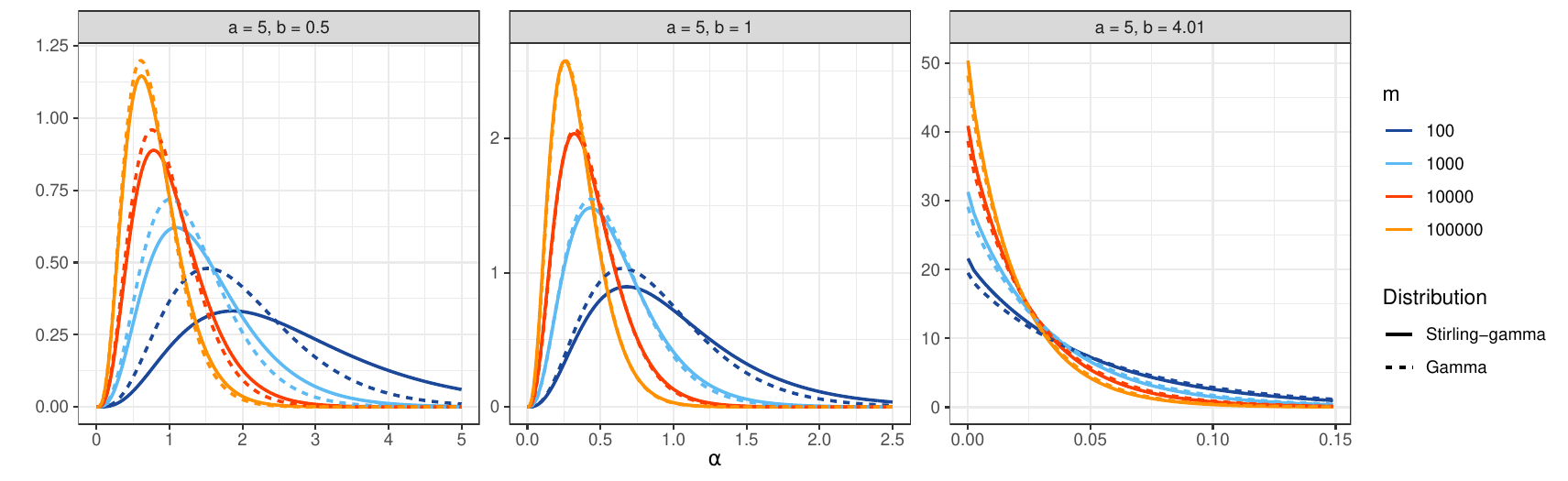}
    \caption{Probability density function of a Stirling-gamma $\mathrm{Sg}(a,b,m)$, depicted by the solid lines, and a $\mathrm{Ga}(a-b, b \log{m})$, indicated by the dashed lines, for varying values of $m$ and $b$, and $a = 5$.}
    \label{fig:Stirling_and_gamma}
\end{figure}

The density function of a Stirling-gamma is proper, namely $\mathcal{S}_{a, b, m} < \infty$, if only if $1 < a/b < m$, as shown in the Supplementary material. Interestingly, the normalizing constant $\mathcal{S}_{a, b, m}$ is the key to calculating the moments of the distribution, which are obtained as follows. 
\begin{proposition}\label{pro:moments}Let $\alpha\sim \mathrm{Sg}(a,b,m)$ and suppose that $0 < s < m b - a$. Then
$$\mathds{E}(\alpha^s) = \frac{\mathcal{S}_{a + s, b, m}}{\mathcal{S}_{a, b, m}}.$$
\end{proposition}
When  $s > mb - a$, instead, then one has $\mathds{E}(\alpha^s) = \infty$. In general, explicit analytic expressions for the moments are not available. One possibility is to approximate $\mathcal{S}_{a, b,m}$ and, consequently, $\mathds{E}(\alpha^s)$ via Monte Carlo integration since samples from the Stirling-gamma can be drawn efficiently; see the Supplementary material. Alternatively, when $m$ is large, we have that $\mathds{E}(\alpha) = \mathcal{S}_{a+1, b,m}/\mathcal{S}_{a, b,m}$ is roughly equal to $(a/b-1)/\log{m}$ and that $\mathcal{S}_{a, b,m} \approx (b\log{m})^{a-b}/\Gamma(a-b)$ by means of Proposition~\ref{pro:limit_gamma}. Also, in the special instance where $a, b\in \mathds{N}$, we can express $\mathcal{S}_{a, b,m}$ analytically as an alternating sum of logarithms, as shown in Theorem S2 in the Supplementary material. 

\subsection{Random partitions via Stirling-gamma priors}\label{sec:StirProcess}
When the precision parameter of a Dirichlet process follows a Stirling-gamma distribution $\alpha \sim \mathrm{Sg}(a, b, m)$, we have a \emph{Stirling-gamma process}. As described in Section~\ref{subsec:Gibbs}, this is a member of the Gibbs-type family with $\sigma = 0$. Thus, the associated exchangeable partition probability function is readily available from the results of \citet{Gnedin2005}. 
\begin{theorem}\label{theo:EPPF_SGP}
The exchangeable partition probability function of a Stirling-gamma process with $\alpha \sim \mathrm{Sg}(a, b, m)$ is 
\begin{equation*}
\mathds{P}(\Pi_n = \{C_1, \ldots, C_k\}) = \frac{\mathscr{V}_{a,b,m}(n, k)}{\mathscr{V}_{a,b,m}(1, 1)}\prod_{j=1}^k (n_j - 1)!,
\end{equation*}
where the coefficients are equal to
$$
\mathscr{V}_{a,b,m}(n, k) = \int_{\mathds{R}_+} \frac{\alpha^{a+k-1}}{\{(\alpha)_m\}^b(\alpha)_n}\mathrm{d}\alpha.
$$
\end{theorem}
It is easy to see that the Gibbs-type coefficients are $V_{n, k} = \mathscr{V}_{a,b,m}(n, k)/\mathscr{V}_{a,b,m}(1, 1)$, for $k = 1, \ldots, n$ and $n \geq 1$, with $V_{1,1}= 1$, and that the forward recursion is satisfied since $\mathscr{V}_{a,b,m}(n, k) = n \mathscr{V}_{a,b,m}(n+1, k) + \mathscr{V}_{a,b,m}(n+1, k+1)$.

Unlike other mixtures of Gibbs-type priors with $\sigma=0$, the Stirling-gamma process has the interesting property of admitting explicit analytical forms for the recursive coefficients. Relying on similar reasoning as the one used for prior coefficients in Definition~\ref{def:StirlingGammaPdf}, we can write $\mathscr{V}_{a, b, m}(n, k)$ when $a, b \in \mathds{N}$ in terms of the polynomial functions. In particular, given the variables  $x_1, \ldots, x_s$ for $s\geq 1$, we introduce the quantities 
\begin{equation}\label{eq:Pfunction2}
\mathscr{S}_{b, j}(x_1, \ldots, x_b) =  \sum_{s=1}^{b} \frac{B_{b-s}(x_{1}, \ldots, x_{b-s})}{(b-s)!}\phi_s(j), \quad
\phi_s(j) =
\begin{cases}
-\log{j}, &s = 1,\\
j^{1-s}/(s-1) &s>1, \\
\end{cases}
\end{equation}
where $B_0(x_0) = 1$ and $B_s(x_{1}, \ldots, x_s)$ is of complete exponential Bell polynomials. See \citet{Charalambides_2005} and the Supplementary material for details. 
Then, the following result holds.
\begin{theorem}\label{theo:NormConst_posterior_rev}
Let $a, b \in \mathds{N}$ and $m \geq 2$, and call $M = \min\{n,m\}$ and $\ell = |n-m|$. Then it holds that
\begin{equation*}\label{eq:NormConst_posterior}
\begin{aligned}
\mathscr{V}_{a,b,m}(n, k)  =  \sum_{j=1}^{M-1}   &(-1)^{\bar{k} - j(b+1)}\frac{j^{\bar{k}}}{\{\Gamma(j)\Gamma(M - j)\}^{b+1}(M - j)_\ell} \mathscr{S}_{b+1, j}(g_{j, 1}, \ldots, g_{j, b +1})\\  
&+ \sum_{i=0}^{\ell -1}(-1)^{\bar{k}+i} \frac{(M+i)^{\bar{k}}}{\Gamma(i + 1)\Gamma(\ell  - i)\{(i+1)_{M-1}\}^{b+1}}\log(M+i),
\end{aligned}
\end{equation*}
where $\bar{k} = a+k-b-2$, are defined in equation~\eqref{eq:Pfunction2} and $g_{j, s}$ are functions of generalized harmonic numbers $H_{j, s} = \sum_{i=1}^j 1/i^s$, equal to
$$g_{j, s} = -(a+k-1)\frac{(s-1)!}{j^{s}} - (s-1)!\{bH_{M-j-1,s} - (b+1)H_{j,s} + bH_{M-j-\ell -1, s}\}.$$ 
\end{theorem}
The implication of Theorem~\ref{theo:NormConst_posterior_rev} is that, together with Theorem S2 in the Supplementary material, it is possible to express $V_{n, k}$ as ratios of alternating sums of logarithms after noticing that $\mathscr{V}_{a,b,m}(1, 1) = \mathcal{S}_{a, b, m}$.

\subsection{Number of clusters in Stirling-gamma processes}

By being a genuine Gibbs-type prior, the Stirling-gamma process admits an urn scheme representation of the form in equation~\eqref{eq:Gibbs_scheme}. In particular, the latent parameters $(\theta_n)_{n \ge 1}$ abide the following generative mechanism: 
\begin{equation}\label{eq:SGDP_scheme}
    \mathds{P}(\theta_{n+1} \in A \mid \!\theta_1, \ldots, \theta_n) = \frac{\mathscr{V}_{a,b,m}(n\!+\!1, k\!+\!1)}{\mathscr{V}_{a,b,m}(n, k)} P_0(A) \! +\! \frac{\mathscr{V}_{a,b,m}(n\!+\!1, k)}{\mathscr{V}_{a,b,m}(n, k)} \sum_{j = 1}^k n_j\delta_{\theta_j^*}(A),
\end{equation}
for $n\geq 1$ and for every $A \in \mathscr{B}(\Theta)$. The fundamental difference between the predictive scheme in equation~\eqref{eq:SGDP_scheme} and the one arising from the generic distribution in~\eqref{eq:eppf_mixtDp} lies in the fact that the hyperparameters of the Stirling-gamma prior are interpretable in terms of the induced number of clusters in the latent partition. We elucidate this with the following key result. 
\begin{theorem}\label{theo:MarginalKm}
Let $\alpha\sim \mathrm{Sg}(a,b,m)$ and $\mathcal{D}_{a, b, m} = \mathds{E}\{\sum_{i=0}^{m-1} \alpha^2/(\alpha + i)^2\}$. The number of clusters $K_m$ obtained from the first $m$ random variables $\theta_1, \ldots, \theta_m$ generated from the predictive scheme in equation~\eqref{eq:SGDP_scheme} is distributed as
\begin{equation}\label{eq:K_m}
\mathds{P}(K_m = k) = \frac{\mathscr{V}_{a, b, m}(m, k)}{\mathscr{V}_{a,b, m}(1, 1)}|s(m, k)|, 
\end{equation}
for $k = 1,\ldots, m$, with mean and variance equal to
$$
\mathds{E}(K_m) = \frac{a}{b}, \qquad \mathrm{var}(K_m) = \frac{b+1}{b}\left(\frac{a}{b} - \mathcal{D}_{a, b, m}\right).
$$
\end{theorem}
In Section~\ref{subsec:robust} we further show that $\mathcal{D}_{a, b, m} \approx 1$ for $m$ large enough. The above statement suggests that hyperparameters $a$, $b$ and $m$ have an important meaning: when $\alpha \sim \mathrm{Sg}(a, b, m)$, the first $m$ statistical units $\{1, \ldots, m\}$ arising from the Stirling-gamma process identify $a/b$ clusters on average, with variance inversely related to $b$. For this reason, we can refer to $m$ as a hypothetical \emph{reference sample size}, $a/b$ as a \emph{location}, and $b$ as a \emph{precision}.   
Theorem~\ref{theo:MarginalKm} also provides an explicit motivation for why the hyperparameters of the Stirling-gamma must satisfy $1<a/b<m$ as in Definition~\ref{def:StirlingGammaPdf}: having $a/b = 1$ is equivalent to having $\mathds{E}(K_m) = 1$, which corresponds to a Dirichlet process where $\alpha \to 0$. On the contrary, setting $a/b = m$ leads to $\mathds{E}(K_m) = m$, meaning that every observation identifies a new cluster. This is the case of a Dirichlet process where $\alpha \to \infty$. Setting $1<a/b<m$ avoids both degenerate behaviors. 

The results in Theorem~\ref{theo:MarginalKm} hold exclusively at the $m\text{th}$ sample from the Stirling-gamma process. For arbitrary \emph{fixed} values of $a, b$ and $m$ the distribution of the number of cluster $K_n$ is given in equation~\eqref{eq:Kn_dist}, whose moments are not available in closed form. It is well known that the number of clusters $K_n$ arising at a generic $n\text{th}$ draw from the predictive scheme in equation~\eqref{eq:SGDP_scheme} maintains the logarithmic divergence typical of the Gibbs-type processes with $\sigma = 0$.  This is because $K_n/ \log{n} \to \alpha \sim \mathrm{Sg}(a, b,m)$ in distribution as $n\to\infty$, as discussed in \citet{Pitman1996}.   
On the other hand, one key aspect of Theorem~\ref{theo:MarginalKm} is that the expectation of the number of clusters among $\theta_1,\dots,\theta_m$, obtained from equation~\eqref{eq:SGDP_scheme}, is \emph{independent} of $m$. Indeed, it will be shown in Section~\ref{subsec:robust} that the distribution of $K_m$ provided in equation~\eqref{eq:K_m} converges to a finite discrete random variable as $m \rightarrow \infty$. This is a consequence of Proposition~\ref{pro:limit_gamma} and the diverging nature of $K_n$ discussed above: while $K_n$ diverges at a logarithmic rate $\alpha\log{n}$, the Stirling-gamma prior makes $\alpha$ approach zero logarithmically in $m$. Hence, when $m = n$ and both $m, n \to \infty$, the divergence to infinity and the convergence to zero happen at the same rate, which implies that the resulting random variable $\alpha \log{m}$ approaches a finite quantity instead of diverging. 

\subsection{Robustness properties}\label{subsec:robust}
In this Section, we investigate the behavior of the number of clusters of the Stirling-gamma process under a large reference sample size.  
Interestingly, if $m$ itself is chosen large, we are able to show that $K_m$ approaches a well-known distribution. 

\begin{theorem}\label{theo:NegBin}
Under the same assumptions of Theorem~\ref{theo:MarginalKm}, the following convergence in distribution holds:
$$
K_m \to K_\infty, \quad K_\infty \sim 1 + \mathrm{Negbin}\left(a - b,\frac{b}{b + 1}\right),  \quad m\to \infty.
$$
\end{theorem}
In the above Theorem, $\mathrm{Negbin}(r,q)$ denotes a negative binomial distribution with mean $r(1-q)/q$ and variance $r(1-q)/q^2$. As such, it holds that 
$$
\mathds{E}(K_\infty) = \frac{a}{b}, \qquad \mathrm{var}(K_\infty) = \frac{b+1}{b}\left(\frac{a}{b} - 1\right).$$
Hence, the quantity $\mathcal{D}_{a, b, m}$ defined in Theorem~\ref{theo:MarginalKm} converges to one when $m\to \infty$. Thus, Theorem~\ref{theo:NegBin} provides a reliable approximation for the prior distribution of the number of clusters. The same result is maintained when $\alpha \sim \mathrm{Ga}(a-b, b\log{m})$. This should not come as a surprise considering the asymptotic equivalence discussed in  Proposition~\ref{pro:limit_gamma}.

In view of Theorem~\ref{theo:NegBin}, it is natural to draw a comparison between the Stirling-gamma process and the Dirichlet process. To mimic the behavior of $\alpha$ under a Stirling-gamma prior, we study the number of clusters from a Dirichlet process at the reference sample size $m$ when $\alpha =  \lambda/\log{m}$, with $\lambda >0$ being a positive constant. The large $m$ behavior is illustrated in the next Proposition, where $\mathrm{Po}(\lambda)$ denotes the Poisson distribution with mean $\lambda$.

\begin{proposition}\label{pro:poisson}
Let $\theta_1, \ldots, \theta_m$ be the first $m$ realizations from a Dirichlet process, obtained by setting $V_{n, k} = \alpha^k/(\alpha)_n$ and $\sigma = 0$ in equation~\eqref{eq:Gibbs_scheme}. If $\alpha = \lambda/\log{m}$ for some $\lambda >0$, then the following convergence in distribution holds:
$$K_m \to K_\infty, \qquad K_\infty \sim 1 + \mathrm{Po}(\lambda), \qquad m\to\infty.$$
\end{proposition}
Similar Poisson-type behaviors for the number of clusters in the Dirichlet process have already been shown in the literature. See for example Proposition 4.8 in \citet{ghosal_van_der_vaart_2017}. Theorem~\ref{theo:NegBin} and Proposition~\ref{pro:poisson} suggest a theoretical reason for why a Dirichlet process with random precision is more flexible than the fixed precision counterpart. When $\alpha$ is kept fixed and sufficiently small, the number of clusters is approximately distributed as a Poisson, whose mean and variance are uniquely controlled by one parameter. 
On the contrary, choosing a Stirling-gamma prior with large $m$ induces an approximately negative-binomial prior for $K_n$, leading to much greater robustness to the prior expectation for $K_n$, as illustrated in Figure~\ref{fig:figure1}.

\section{Conjugate inference under Stirling-gamma priors}\label{sec:inference}
In this Section, we illustrate how the Stirling-gamma distribution has the further important property of being \emph{conjugate} to the law of the partition of the Dirichlet process. As we show in Proposition~\ref{pro:conjugacy}, this happens when the reference sample size $m$ is set equal to the number of data points $n$ in equation~\eqref{eq:dp_eppf}. 
\begin{proposition}\label{pro:conjugacy}
Suppose we observe a partition $\Pi_n$ distributed according to the Dirichlet process in~\eqref{eq:dp_eppf} and let $\alpha\sim \mathrm{Sg}(a, b, n)$. Then, $(\alpha\mid\Pi_n = \{C_1, \ldots, C_k\}) \sim \mathrm{Sg}(a + k, b + 1, n)$.
\end{proposition}
The same result can be derived by conditioning on $K_n = k$ alone as in \citet{Escobar_West_1995} because of its sufficiency for $\alpha$. The above conjugacy simplifies computations when sampling from the posterior distribution in a Dirichlet process mixture model with random precision, which in the case of the gamma prior requires a data augmentation step. Under the conjugate Stirling-gamma prior, elicitation is straightforward by virtue of Theorems~\ref{theo:MarginalKm} and~\ref{theo:NegBin}. Thus, one can transparently tune the Stirling-gamma prior by leveraging upon information available on the clustering structure of the $n$ observations through choices of $a$ and $b$. 

The existence of the conjugate Stirling-gamma prior follows directly from the results of \citet{Diaconis_Yilvisaker_1979} for natural exponential families, which the partition law of the Dirichlet process is a member of. Nevertheless, the prior dependence on $n$ has some important consequences on the process, which must be handled with care. In particular, while the distribution in Theorem~\ref{theo:EPPF_SGP} remains the one of a finitely exchangeable product partition model, the Gibbs-type recursion characterizing the coefficients $V_{n, k}$ no longer holds. Namely, 
$V_{n,k} \neq nV_{n+1, k} + V_{n+1, k+1}$. This breaks the predictive scheme of equation~\eqref{eq:SGDP_scheme}, causing the sequence to lose the \emph{projectivity} property typical of species sampling models \citep{Lee2013}. In other terms, the distribution in Theorem~\ref{theo:EPPF_SGP} under $n$ observations does not coincide with the one obtained by marginalizing out the $(n+1)$th sample from the same distribution under $n+1$ data points. This is a limitation when one is interested in extrapolating inferences from a sample to the general population, but less so on clustering problems where out-of-sample predictions are not the main focus \citep{Betancourt_2020}. 

The lack of projectivity of the sequence under $m = n$ is less relevant in settings where $n$ plays the role of the dimension of the data rather than the number of observed data points. We illustrate this by introducing the following \emph{population of partitions} framework. Let $\Pi_{n, 1}, \ldots, \Pi_{n, N}$ denote $N$ independent and identically distributed realizations of a random partition of the same units $\{1, \ldots, n\}$ from an exchangeable partition probability function. If each partition is from a Dirichlet process with precision $\alpha$, then we have
\begin{equation}\label{eq:RepEPPF}
\mathds{P}(\Pi_{n, s} = \{C_{1, s}, \ldots, C_{k_s, s}\} \mid \alpha) = \frac{\alpha^{k_s}}{(\alpha)_n}\prod_{j = 1}^{k_s} (n_{j, s} - 1)!,  \quad (s = 1, \ldots, N),
\end{equation}
where $n_{j, s} = |C_{j, s}|$ is the number of elements in the $j$th cluster $C_{j, s}$ within the $s$th partition, and $k_s$ is the associated number of clusters. The model in equation~\eqref{eq:RepEPPF} is suitable for instances where, for example, we measure the interactions among the same $n$ nodes of a network multiple times. Similar data often occur in neuroscience studies, where the same $n$ brain regions are scanned for $N$ different individuals \citep{Durante_2017}, or in ecology, where the interactions among $n$ species are recorded for $N$ days \citep{Mersch2013}. The inferential goal of model~\eqref{eq:RepEPPF} is to retrieve the network-specific partition through a shared Dirichlet process precision parameter. Then, the following Theorem holds.
\begin{theorem}\label{theo:Conjugacy}
Let $\Pi_{n, 1}, \ldots, \Pi_{n, N}$ be independent and identically distributed realizations from equation~\eqref{eq:RepEPPF}. If $\alpha \sim \mathrm{Sg}(a,b,n)$, then $$(\alpha \mid \Pi_{n, 1}, \ldots, \Pi_{n, N}) \sim \mathrm{Sg}\left(a + \sum_{s=1}^N k_s, b + N, n\right).$$
\end{theorem}
It is straightforward to notice that Proposition~\ref{pro:conjugacy} is retrieved by letting $N=1$ in the above. In light of Theorem~\ref{theo:Conjugacy}, we can also derive the classic Bayesian decomposition of the posterior mean as a weighted average between the observed data and the prior. Recall that $\mathds{E}(K_n\mid \alpha) = \sum_{i=0}^{n-1} \alpha/(\alpha + i)$ is the conditional mean for the number of clusters generated by a Dirichlet process over partitions of the units $\{1,\ldots, n\}$, and that $\mathds{E}(K_n) = \mathds{E}\{\mathds{E}(K_n\mid \alpha)\} =a/b$ thanks to the law of the iterated expectation. Then, the next Proposition holds.

\begin{proposition}\label{pro:PosteriorAlpha}
Under the same setting of Theorem~\ref{theo:Conjugacy}, we have 
$$
\mathds{E}\left(\sum_{i=0}^{n-1} \frac{\alpha}{\alpha + i}\mid \Pi_{n, 1}, \ldots, \Pi_{n, N} \right) = \frac{b}{b+N} \frac{a}{b} + \frac{N}{b+N} \bar{k},
$$
where $\bar{k} = N^{-1}\sum_{s=1}^{N} k_s$ is the average number of clusters observed across the partitions.
\end{proposition}
The above statement is a direct consequence of the conjugacy of the Stirling-gamma prior under $m = n$. See \citet{Diaconis_Yilvisaker_1979} and the Supplementary material for details.

\begin{remark} The results in Proposition~\ref{pro:conjugacy} and Theorem~\ref{theo:Conjugacy} are particularly useful when applied to mixture modeling. In this respect, \citet{JMLR:v15:miller14a} showed that Dirichlet process mixtures with fixed precision $\alpha$ are \emph{inconsistent} in retrieving the correct number of clusters when data are generated from a finite mixture, even when mixture kernels are correctly specified. However, the recent contribution of \citet{Ascolani_2022} counters this argument and proves that choosing $\alpha$ as random yields consistency if some strong assumptions hold on both $\pi(\alpha)$ and the data-generating process. In particular, they enlist the gamma distribution as a valid choice for $\pi(\alpha)$. We show in the Supplementary material that the Stirling-gamma prior $\alpha \sim \mathrm{Sg}(a, b, n)$  satisfies only two out of the three assumptions required by \citet{Ascolani_2022}, despite its asymptotic equivalence with gamma highlighted in Proposition~\ref{pro:limit_gamma}. Since we cannot use their proof in our case, we will address consistency in simulated examples in Section~\ref{sec:simulations}.
\end{remark}


\section{Computational aspects and sampling strategies}\label{sec:computations}
There exists a rich variety of algorithms to sample from the posterior distribution of the mixture model in~\eqref{eq:MixtModel}. For Gibbs-type processes, one popular approach lies in the class of marginal samplers, which rely on the sequential predictive scheme of equation~\eqref{eq:Gibbs_scheme}. See \citet{Escobar_West_1995, Neal_2000} for examples. In light of its hierarchical construction, inference under the Stirling-gamma process mixture model can be performed under the same marginal scheme of the Dirichlet process, with an additional sampling step for $\alpha$. Such a step is provided by Proposition~\ref{pro:conjugacy} or Theorem~\ref{theo:Conjugacy} depending on the setting. In both cases, the conditioning is with respect to the last sampled partition at the given iteration. 
\begin{algorithm}[t]
\caption{Rejection sampler for the Stirling-gamma distribution}\label{algo:Stirgamma_sampler}
\begin{algorithmic}
\State \small{\textbf{Input}}: parameters $a, b>0$ and $m\in\mathds{N}$ with $1<a/b<m$.
\If{$a-b \geq 1$}
\State Let $M_u = \max_{\alpha \geq 0} S(\alpha)$ and $M_v = \max_{\alpha \geq 0} \alpha^2S(\alpha)$, with $S(\alpha) = \alpha^{a-1}/\{(\alpha)_m\}^b$
\State Sample $(u, v)$ uniformly in  $[0, M_u^{1/2}]\times [0, M_v^{1/2}]$
\If{$2\log{u} \leq \log S(v/u)$}
\State Accept $\alpha = v/u$ 
\Else \State Reject
\EndIf
\Else
\State Let $r = \Gamma(m)^{1/(m-1)}$ and $A(\alpha)= (\alpha+r)^{b(m-1)}/(\alpha+1)^b_{m-1}$
\State Sample $x \sim \mathrm{Be}(a-b, mb-a)$ and set $y = r x/(1-x)$, and sample $u$ uniformly in  $[0, 1]$
\If{$\log{u} \leq \log A(y)$} 
\State Accept $\alpha = y $
\Else \State Reject
\EndIf
\EndIf
\State \small{\textbf{Output}}: a sample from $\alpha \sim\mathrm{Sg}(a, b, m)$. 
\end{algorithmic}
\end{algorithm}

To facilitate the adoption of our prior in applications, we hereby provide a strategy to draw random samples from the Stirling-gamma distribution $\alpha\sim\mathrm{Sg}(a,b,m)$, which can also be utilized for the conjugate posteriors in Section~\ref{sec:inference} when fitting Dirichlet process mixture models. Unfortunately, a rejection sampler using the gamma $\mathrm{Ga}(a-b, b\log{m})$ in Proposition~\ref{pro:limit_gamma} as the proposal is not feasible because the Stirling-gamma has heavier tails; see the Supplementary material. As such, an ideal proposal must be itself a heavy-tailed distribution from which sampling is relatively easy. We find a suitable candidate to be the \emph{generalized beta prime distribution}, denoted as $\alpha \sim \mathrm{BeP}(a_0, b_0, r)$ and whose density is 
\begin{equation}\label{eq:BetaPrimePdf}
\pi_{\mathrm{BeP}}(\alpha) = \frac{(\alpha/r)^{a_0 - 1}(1 + \alpha/r)^{-a_0 - b_0}}{r\beta(a_0, b_0)}
\end{equation}
with $\alpha > 0$ and $\beta(a_0, b_0) = \Gamma(a_0)\Gamma(b_0)/\Gamma(a_0 + b_0)$ denoting the Beta function. Its specific advantage is that samples can be drawn by letting $\alpha = r x/(1-x)$  with $x \sim \mathrm{Be}(a_0, b_0)$. Moreover, under suitable for values for $a$, $b$, and $r$, the density in equation~\eqref{eq:BetaPrimePdf} ``covers'' the Stirling-gamma density uniformly, as the next result shows.
\begin{proposition}\label{pro:BetaPrime_ratio}
    Let $\pi_\mathrm{Sg}(\alpha)$ denote the density of  $\alpha\sim \mathrm{Sg}(a, b, m)$, and call $\pi_{\mathrm{BeP}}(\alpha)$ the density of  $\alpha \sim \mathrm{BeP}(a-b, mb-a, r)$ as in equation~\eqref{eq:BetaPrimePdf}, with $r = \Gamma(m)^{1/(m-1)}$. Then,
    $$
    \frac{\pi_\mathrm{Sg}(\alpha)}{\pi_\mathrm{BeP}(\alpha)} \leq \frac{\beta(a-b, mb-a)}{r^{mb-a}\mathcal{S}_{a,b,m}}  = M < \infty.
    $$
\end{proposition}

The quantity $M$ is a finite upper bound. Hence, the generalized beta prime can be used as a proposal in an accept-reject algorithm \citep{Devroye86}, where the acceptance probability is $1/M$ and $A(\alpha) = \pi_\mathrm{Sg}(\alpha)/\{M\pi_\mathrm{BeP}(\alpha)\} = (\alpha+r)^{b(m-1)}/(\alpha+1)^b_{m-1}$ is the acceptance function. We find such a strategy particularly effective when $a-b<1$, which corresponds to the case when the density in Definition~\ref{def:StirlingGammaPdf} is unbounded near the origin as $\alpha\to0$. However, when $a-b\geq 1$, the quantity $1/M$ becomes particularly low. Nevertheless, in this case, the density admits a finite maximum. Hence, we rely on the ratio of uniforms method, which can be adapted to any distribution for which the density $f(x)$ and the function $x^2f(x)$ can be maximized; refer to \citet{Devroye86} for a thorough description.

\begin{table}[t]
\begin{adjustbox}{max width=1\textwidth,center}
\begin{tabular}{llcccc|cccc}
\toprule
& & \multicolumn{4}{c}{$m = 100$} & \multicolumn{4}{c}{$m = 1000$} \\
&   & $a = 2$ & $a = 3$ & $a = 10$ & $a = 15$ & $a = 2$ & $a = 3$ & $a = 10$ & $a = 15$\\
\midrule
$a-b \geq 1$& $b = 0.2$ & 0.756 & 0.701 & 0.544 & 0.594 & 0.742 & 0.668 & 0.425 & 0.358\\
& $b = 1$ & 0.679 & 0.724 & 0.445 & 0.377 & 0.680 & 0.717 & 0.419 & 0.346\\
& $b = 1.5$ &  & 0.754 & 0.446 & 0.372 &  & 0.752 & 0.427 & 0.349\\
& $b = 5$ &  &  & 0.528 & 0.394 & &  & 0.523 & 0.386\\
\midrule
  & & $a = 0.2$ & $a = 0.6$ & $a = 0.7$ & $a = 1$ & $a = 0.2$ & $a = 0.6$ & $a = 0.7$ & $a = 1$\\
  \midrule
$a-b < 1$ & $b = 0.1$ & 0.949 & 0.788 & 0.760 & 0.678 & 0.911 & 0.638 & 0.593 & 0.458\\
& $b = 0.2$ &  & 0.799 & 0.757 & 0.655 &  & 0.683 & 0.622 & 0.476\\
& $b = 0.5$ &  & 0.940 & 0.883 & 0.733 &  & 0.907 & 0.822 & 0.609\\
& $b = 0.6$ &  &  & 0.938 & 0.775 &  &  & 0.905 & 0.670\\
\bottomrule
\end{tabular}
\end{adjustbox}
\caption[Acceptance probabilities for Algorithm \ref{algo:Stirgamma_sampler} under varying $a$, $b$ and $m$.]{Acceptance probabilities for Algorithm \ref{algo:Stirgamma_sampler} under varying $a$, $b$ and $m$. Values are obtained by averaging the acceptance rate obtained in 1000 trials of Algorithm~\ref{algo:Stirgamma_sampler} under 100 replicates. Standard deviations were all around $0.01$ and therefore are omitted from the table. Empty cells indicate when $1<a/b<m$  and the $a-b$ condition is violated }\label{tab:accept}%
\end{table}

Algorithm~\ref{algo:Stirgamma_sampler} summarizes the sampler of the Stirling-gamma. To ease reproducibility, we implement it in the \texttt{R} package \texttt{ConjugateDP} via the function \texttt{rSg}. Table~\ref{tab:accept} reports the acceptance rates for the samplers above for selected values of $a$, $b$, and $m$. We see that, when $a-b\geq 1$, the rates range between $0.3$ and $0.7$, which is fairly large considering that the ratio of uniforms is a method that is not tailored specifically to the Stirling-gamma. Instead, when $a-b <1$, the rates are much larger and range between $0.4$ and $0.95$. This is due to the high similarity between the generalized beta prime distribution and the Stirling-gamma.

\section{Empirical demonstrations}\label{sec:simulations}
\subsection{Posterior comparison between Stirling-gamma and gamma}\label{subsec:Gamma_vs_Sg}
The similarity between the gamma and the Stirling-gamma highlighted in Proposition~\ref{pro:limit_gamma} naturally poses the question of whether the two distributions lead to different posteriors for $\alpha$ and $K_n$ when fitting Dirichlet process mixture models. In a foundational paper, \citet{Escobar_West_1995} first proposed to adopt $\alpha\sim \mathrm{Ga}(2,4)$ as a default. They also present a data-augmentation Gibbs sampling scheme based on beta-distributed random variables to sample from the full conditional of $\alpha$. In light of our result in Proposition~\ref{pro:conjugacy}, where such augmentation is not required, we investigate differences in the posteriors induced by the two priors.

We consider a dataset of $n = 500$ observations generated independently from a mixture of $k^*$ equally weighted bivariate normal components, with variance-covariance matrix equal to $\textrm{diag}\{0.15, 0.15\}$. In particular, we study three settings: one with $k^*=2$, with means $(-1,-1)$ and $(1,1)$; one with $k^*=4$ with means $(-1,-1)$, $(1,-1)$, $(-1,1)$ and $(1,1)$; and finally one with $k^*=6$ and two additional means $(3,0)$ and $(0,3)$, respectively. For each choice of $k^*$, we randomly generate 20 different datasets. 
In all cases, we let $\theta_i = (\mu_i, \Sigma_i)$ and $f(x\mid\theta_i) = N(x; \mu_i, \Sigma_i)$ in equation~\eqref{eq:MixtModel}, with $N(x, \mu, \Sigma)$ denoting a normal distribution with mean $\mu \in\mathds{R}^2$ and variance-covariance matrix $\Sigma\in \mathds{R}^{2\times 2}$. Our prior $\mathscr{Q}$ is a Dirichlet process with precision parameter $\alpha$ and normal-inverse Wishart baseline distribution $N(\mu; 0, \Sigma/\kappa_0)IW(\Sigma; \nu_0, I)$, with $\nu_0 = \kappa_0 = 2$. Three different choices of priors are considered with respect $\alpha$: (i) a weakly-informative conjugate prior $\alpha \sim \mathrm{Sg}(1, 0.25, n)$; (ii) $\alpha \sim \mathrm{Sg}(4, 1, n)$, which is more informative than case (ii); and (iii) the default gamma prior $\alpha \sim \mathrm{Ga}(2, 4)$ as in \citet{Escobar_West_1995}. In all cases, the induced distribution on $K_n$ has mean $\mathds{E}(K_n) \approx 4$, and we also have $\mathds{E}(\alpha) \approx 1/2$. Hence,  the major difference between (i) and (ii)-(iii) lies in the prior vagueness for $K_n$, whereas (ii) differs from (iii) in terms of right tail heaviness. Inference on the number of clusters in each replicate is performed by running a marginal Gibbs sampler as in Algorithm 3 in \citet{Neal_2000} for 15,000 iterations, discarding the first 5,000 as burn-in.

\begin{figure}[t]
    \centering
    \includegraphics[width = \linewidth]{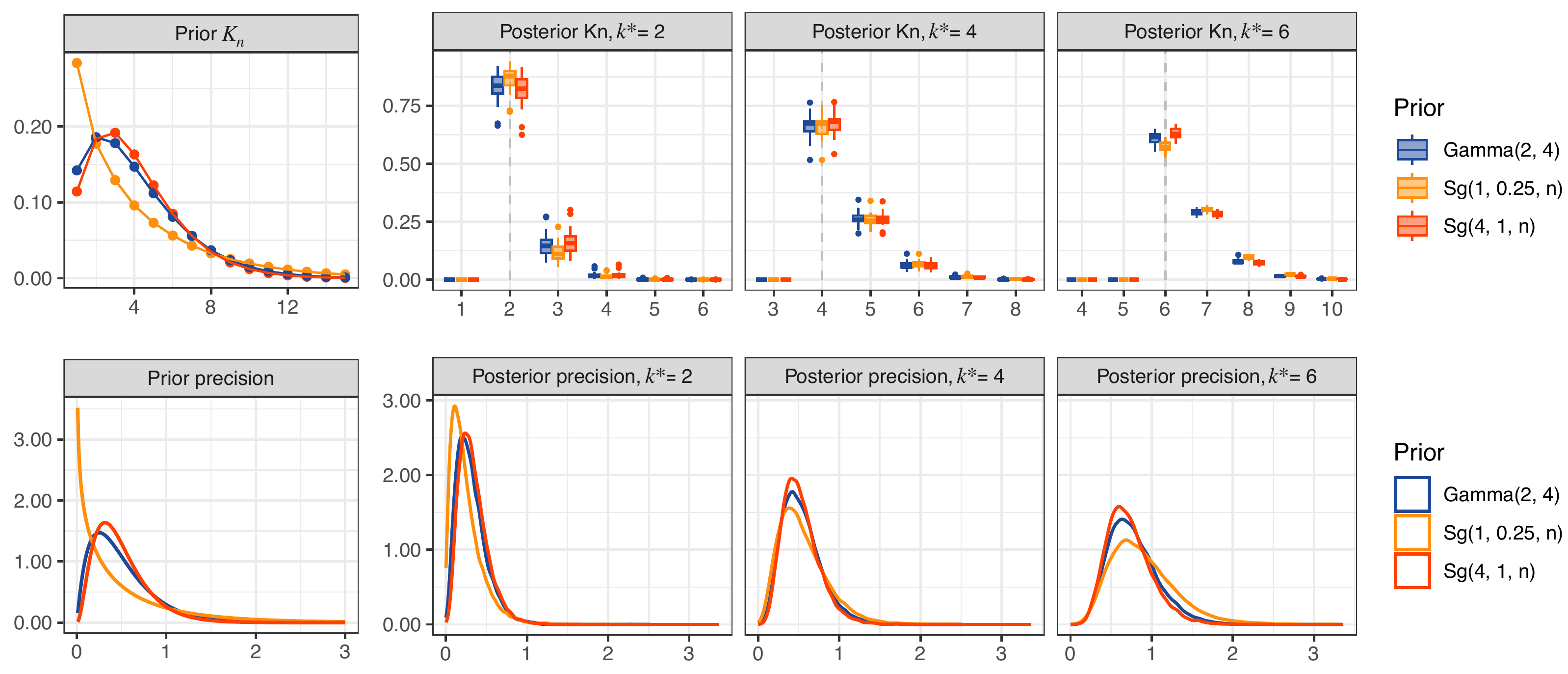}
    \caption{Top panels: probability mass function for the prior $K_n$ (leftmost plot) and posterior probabilities that $K_n=k$ across 20 replicates under the three different priors for $\alpha$. Vertical dashed lines indicate the true number of clusters $k^*$ in each simulation. Bottom panels: priors for $\alpha$ (leftmost plot) and posteriors in one of the replicates. }
    \label{fig:gamma_sg_simulation}
\end{figure}

Figure~\ref{fig:gamma_sg_simulation} displays the results of our simulation. We see that the default gamma prior (iii) and the informative Stirling-gamma (ii) have very similar densities and induce a similar prior for $K_n$. This is expected since both priors are chosen to have mean equal and approximately equal to $1/2$, respectively. The weakly informative Stirling-gamma prior, instead, assigns larger mass to small values of $K_n$. The posteriors, however, resemble one another across all three priors and values for $k^*$.  In particular, all three models assign the largest posterior probability of $K_n$ to the correct value for $k^*$, depicted by the vertical dashed lines. Such probabilities are fairly equal across replicates, indicating that posterior differences between the Stirling-gamma and the gamma are minimal in these settings, irrespective of the heaviness of the right tail and relative flatness of the prior. Minimal differences are also detected in the average posterior effective sample sizes for $\alpha$, which are reported in Table~\ref{tab:effectiveSize}. Specifically, the informative Stirling-gamma appears to have slightly higher effective sample sizes than the gamma counterpart. However, such benefits are marginal since the major computational hurdles of Dirichlet process mixtures lie in the sequential allocation scheme for each observation $i = 1, \ldots, n$, as in equation~\eqref{eq:SGDP_scheme}.

The choice of $\alpha\sim \mathrm{Ga}(2,4)$ is rather conservative, since  $\mathds{E}(\alpha) = 1/2$. However, the induced prior distribution over $K_n$ still varies with $n$. For instance, under the default gamma prior, we have $\mathds{E}(K_{100}) = 3.20$, $\mathds{E}(K_{1000}) = 4.35$, and $\mathds{E}(K_{10000}) = 5.51$.  While Figure~\ref{fig:gamma_sg_simulation} did not detect substantial differences in the posteriors, having a prior that varies with $n$ such as $\mathrm{Sg}(a, b, n)$ could further counterbalance the logarithmic growth of $K_n$ and, in turn, can be beneficial in preventing over-clustering. We will comment further on the matter in the next Section. 

\begin{table}[t]
\begin{adjustbox}{max width=1\textwidth,center}
\begin{tabular}{lccc}
\toprule
\textsc{prior}&  $k^*=2$ & $k^*=4$ & $k^*=6$\\
\midrule
  $\alpha \sim \mathrm{Ga}(2, 4)$ & 8066.35 (581.56) & 7212.07 (521.74) & 7537.74 (394.80) \\
    $\alpha \sim \mathrm{Sg}(4, 1, n)$ & 8536.85 (602.03) & 8107.74 (764.88) & 8354.43 (399.28) \\
    $\alpha \sim \mathrm{Sg}(1, 0.25, n)$ & 7826.41 (812.74) & 6974.12 (582.84) & 7449.69 (411.10) \\
\bottomrule
\end{tabular}
\end{adjustbox}
\caption[]{Average effective sample sizes for $\alpha$ across 20 replicates when estimating a Dirichlet process mixture model on $n = 500$ data points from a mixture of $k^*$ bivariate normals. Standard deviations are indicated in parenthesis.}\label{tab:effectiveSize}%
\end{table}

\subsection{Standard normal Dirichlet process mixture}\label{subsec:normalDPM}

We now empirically investigate how modeling $\alpha$ via the conjugate Stirling-gamma prior can mitigate over-clustering in settings where keeping $\alpha$ fixed leads to inconsistent estimates for the number of clusters. In a famous example, \citet{Miller_2013} first showed that when the data are generated from a single standard normal, namely $X_1, \ldots, X_n\stackrel{\text{iid}}{\sim} N(0, 1)$, the posterior number of clusters arising from a Dirichlet process mixture model \emph{does not} concentrate around one, even if the mixture kernels are correctly specified. In particular, they analyze the behavior of the \emph{standard normal} mixture 
\begin{equation}\label{eq:standardDPM}
X_{i}\mid \theta_i \stackrel{\mathrm{ind}}{\sim} N(\theta_i, 1), \quad \theta_i\mid \tilde{p}\stackrel{\mathrm{iid}}{\sim} \tilde{p}, \quad \tilde{p}\sim \textsc{dp}(\alpha, P_0)
\end{equation}
with baseline density $p_0$ being $N(0, 1)$, and show that $\mathds{P}(K_n = 1 \mid X_1, \ldots, X_n)\to 0$  in probability as $n\to \infty$  when $\alpha=1$. This implies that the model in equation~\eqref{eq:standardDPM} is asymptotically inconsistent with respect to the true number of clusters. Refer to \citet{JMLR:v15:miller14a} for more general examples. 

\begin{figure}[tp]
    \centering
\includegraphics[width = \linewidth]{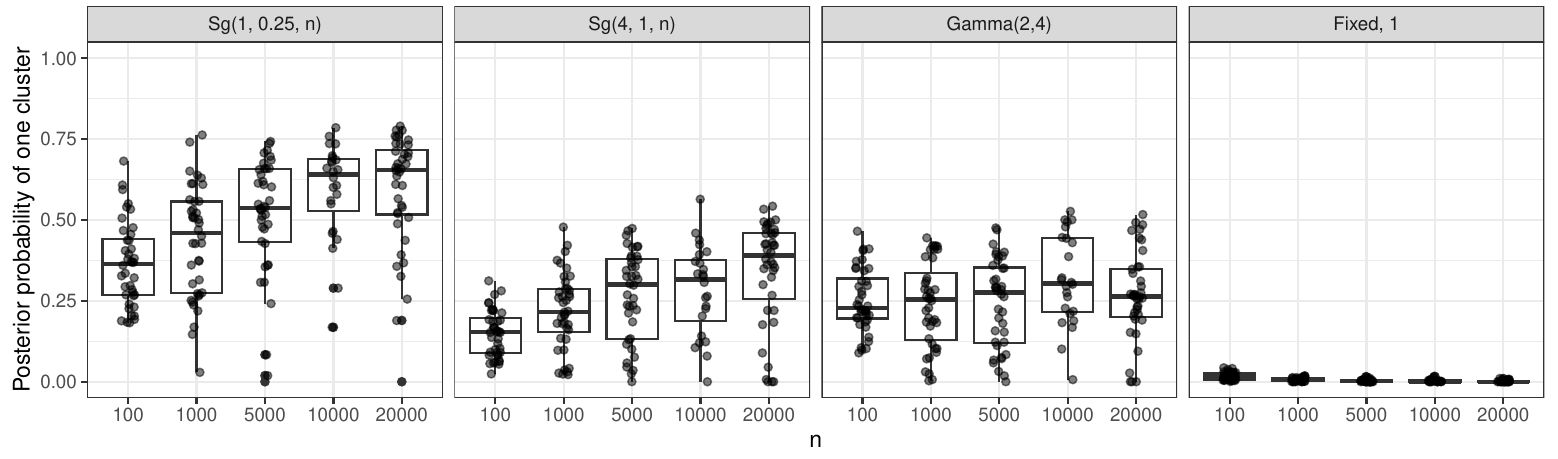}
    \caption{Posterior probability that $K_n=1$ in the standard normal Dirichlet process mixture in equation~\eqref{eq:standardDPM}, when data are independently generated from a single standard normal $X_i \sim N(0,1)$. Each point represents the result of one simulated dataset for each choice of priors over $\alpha$.}
    \label{fig:Consistency_DPM}
\end{figure}

In light of the recent contribution of \citet{Ascolani_2022} mentioned above, we revisit the problem from an empirical perspective by letting $\alpha$ be random. In particular, our aim is to evaluate the behavior of $\mathds{P}(K_n=1\mid X_1, \ldots,X_n)$ for increasingly larger values of $n$, considering again (i) $\alpha \sim \mathrm{Sg}(1, 0.25, n)$; (ii) $\alpha \sim \mathrm{Sg}(4, 1, n)$  (ii);  (iii) $\alpha \sim \mathrm{Ga}(2, 4)$ as in \citet{Escobar_West_1995}; and finally (iv) $\alpha =1$ as studied by \citet{Miller_2013}. Cases (i) and (ii) impose an explicit dependence on the sample size in the prior, which makes them both concentrate at zero progressively as $n$ increases; see Proposition~\ref{pro:limit_gamma}. Case (iii) induces a different prior over $K_n$ based on the value of $n$ instead. We simulate 40 datasets of size $n\in \{100, 1,000, 5,000, 10,000, 20,000\}$ each from a single standard normal, and we estimate the model in equation~\eqref{eq:standardDPM} running a marginal Gibbs sampler for $3,000$ iterations, discarding the first $1,000$ as burn-in. In cases (i) and (ii), we obtain posterior samples from $\alpha$ by applying Algorithm~\ref{algo:Stirgamma_sampler} jointly with Proposition~\ref{pro:conjugacy}, while case (iii) relies again on the augmentation proposed by \citet{Escobar_West_1995}. We report effective sample sizes in the Supplementary material.

Figure~\ref{fig:Consistency_DPM} displays the results. Not surprisingly, $\mathds{P}(K_n=1\mid X_1, \ldots,X_n)$ progressively approaches zero when $\alpha = 1$  \citep{Miller_2013}. Letting $\alpha$ be random, instead, allows the model to learn the precision from the data and, therefore, partially counteracts the tendency of Dirichlet process priors to favor unbalanced partitions \citep{Changwoo_2022}. As such, the resulting posterior does not concentrate away from $K_n=1$. Interestingly, when $\alpha\sim\mathrm{Ga}(2, 4)$, the probability of detecting one cluster appears somehow constant across $n$ across the 40 replicates. The Stirling-gamma priors, instead, show a clear increasing trend in both cases (i) and (ii). This is more evident in case (i) since the induced prior on $K_n$ is vague and, therefore, has a milder impact.  This suggests that adaptively tuning the prior for $\alpha$ with $n$ leads to a more stable behavior than having a one-fit-all default choice such as $\alpha\sim\mathrm{Ga}(2,4)$, and indicates $\alpha\sim \mathrm{Sg}(1, 0.25, n)$ as a good default in applied settings. We conclude by pointing out that, despite the increasing trend, $\mathds{P}(K_n=1\mid X_1, \ldots,X_n)$ is far from concentrating at one in all random settings. Hence, based on the slow logarithmic rate of convergence to 0 discussed in Proposition~\ref{pro:limit_gamma},  we expect that values for $n$ in the order of millions will be needed to see the consistency in simulation.

\section{Application: inferring sub-communities in ant colonies}\label{sec:ants}

We now further illustrate how modeling the precision parameter $\alpha$ via a Stirling-gamma prior in a Dirichlet process mixture as opposed to keeping it fixed yields a more robust estimate of the posterior partition. We specifically consider the problem of detecting community structures in a colony of ant workers by modeling daily ant-to-ant interaction networks via stochastic block models \citep{Nowicki_Snijders_2001}. The data were collected by \citet{Mersch2013} by continuously monitoring six colonies of the ant \emph{Camponotus fellah} through an automated video tracking system over a period of 41 days. Each day yielded a weighted adjacency matrix whose edges contain the number of individual interactions between workers. In this analysis, we model a binary version of the data from days two, four, six, and eight for the fifth colony ($n=149$), where $X_{i, j, s}$ equals one if ant $i$ and ant $j$ in network $s$ interacted more than five times, and zero otherwise. In line with the setting proposed by Theorem~\ref{theo:Conjugacy}, we use the Stirling-gamma process to independently model the $N=4$ latent partitions of the same $n = 149$ ants. Figure~\ref{fig:ant_networks} reports the binary adjacency matrices recording ant interactions. Rows and columns have been sorted according to the three social organization groups retrieved by \citet{Mersch2013}, namely \emph{foragers}, \emph{cleaners}, and \emph{nurses}. The last group is composed of younger individuals and forms a stronger connection with the queen.

\begin{figure}[t]
\centering
\includegraphics[width = \linewidth]{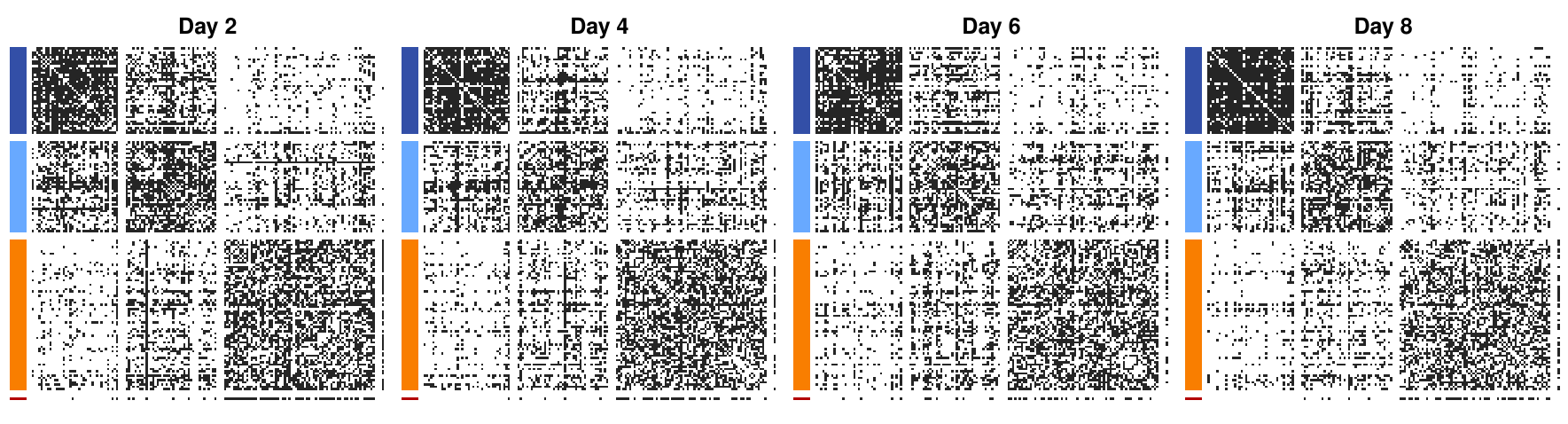}
\caption{Binary ant-to-ant interaction networks in a colony \emph{Camponotus fellah} observed in four different days. Each node is an ant, and black points denote edges. The colors on the left indicate the three groups of workers, namely foragers (dark blue), cleaners (light blue), and nurses (orange). The bottom red node indicates the queen.}
\label{fig:ant_networks}
\end{figure}

In order to perform community detection on each of the four networks, we rely on a stochastic block model formulation. This is a variant of equation~\eqref{eq:MixtModel} of the Introduction, which is best rewritten with the help of auxiliary variables representing cluster assignment as follows. Given a random partition of the nodes $\Pi_{n,s} = \{C_{1, s}, \ldots, C_{k_s, s}\}$ in $s$, call $Z_{i,s}$ an auxiliary variable so that $Z_{i, s} = h$  if the node $i \in C_{h,s}$, for $i = 1, \ldots, n$. The probability of detecting an edge between nodes $i$ and $j$ in network $s$ is specified as
\begin{equation}\label{eq:sbm}
\mathds{P}(X_{i, j,s} = 1 \mid Z_{i, s} = h, Z_{j,s} = h', \nu) = \nu_{h, h',s}, \quad \nu_{h, h',s} \sim \mathrm{Be}(1, 1).
\end{equation}
Here, $\nu_{h, h',s} \in \nu = (\nu_{1,1,1}, \ldots,)$ is the edge probability in the block identified by clusters $C_{h, s}$ and $C_{h',s}$, and $\mathrm{Be}(a_0,b_0)$ is the Beta distribution with mean $a_0/(a_0 + b_0)$.  We assume no node self-relation, thus ignoring the diagonal entries $X_{i, i, s}$. By modeling the latent partition via the Dirichlet process prior $\mathds{P}(\Pi_{n,s}\mid \alpha)$ as in equation~\eqref{eq:dp_eppf}, we can flexibly find a grouping of the nodes with a similar edge distribution and thus infer the number of ant worker communities without pre-specifying an upper bound to the number of clusters. See \citet{Legramanti2022} and references therein for a description of the posterior sampling algorithm.  The specific advantage of this approach is that we can identify richer \emph{sub-communities} than the three large ones described originally.
 
Our intent is to investigate the impact that different choices of $\alpha$ have on each posterior partition from the model in equation~\eqref{eq:sbm}. In particular, we compare a Dirichlet process mixture with (i) $\alpha = 0.4$ and (ii) $\alpha = 18$, against Stirling-gamma processes with (iii)  $\alpha \sim \mathrm{Sg}(0.6, 0.2, 149)$, and (iv)  $\alpha \sim \mathrm{Sg}(8, 0.2, 149)$. The hyperparameters in models (i) and (iii) are chosen such that $\mathds{E}(K_n) = 3$ so as to incorporate the \emph{a priori} knowledge of the three groups described by \citet{Mersch2013}. To check for posterior robustness, cases (ii) and (iv), instead, have $\mathds{E}(K_n) = 40$. As is evident from the leftmost panel of  Figure~\ref{fig:Prior_Post_Kn}, the Stirling-gamma prior is less informative about $K_n$ due to the choice of $b = 0.2$. We obtain the posterior partition in each model by running a collapsed Gibbs sampler as in \citet{Legramanti2022} for $40,000$ iterations, treating the first $10,000$ as burn-in. The full conditional for $\alpha$ in both Stirling-gamma processes is provided by Theorem~\ref{theo:Conjugacy}, setting $N=4$ and $n = 149$. The resulting effective sample size for $\alpha$ in cases (iii) and (iv) is $11,932.22$ and $19,422.59$, indicating good mixing. Figure~\ref{fig:Prior_Post_Kn} plots the posterior distribution for the number of retrieved clusters $K_n$ in each network. We see that there exists a non-negligible difference between the two posteriors where $\alpha$ is kept fixed, leading to under- and over-clustering in cases (i) and (ii), respectively. On the contrary, making $\alpha$ random with a sufficiently vague Stirling-gamma prior retrieves a virtually identical posterior irrespective of the induced prior on $K_n$. This is due to the additional flexibility granted by the Stirling-gamma, which enables the model to infer $\alpha$ from the data. 
\begin{figure}[tb]
\centering
\includegraphics[width = \linewidth]{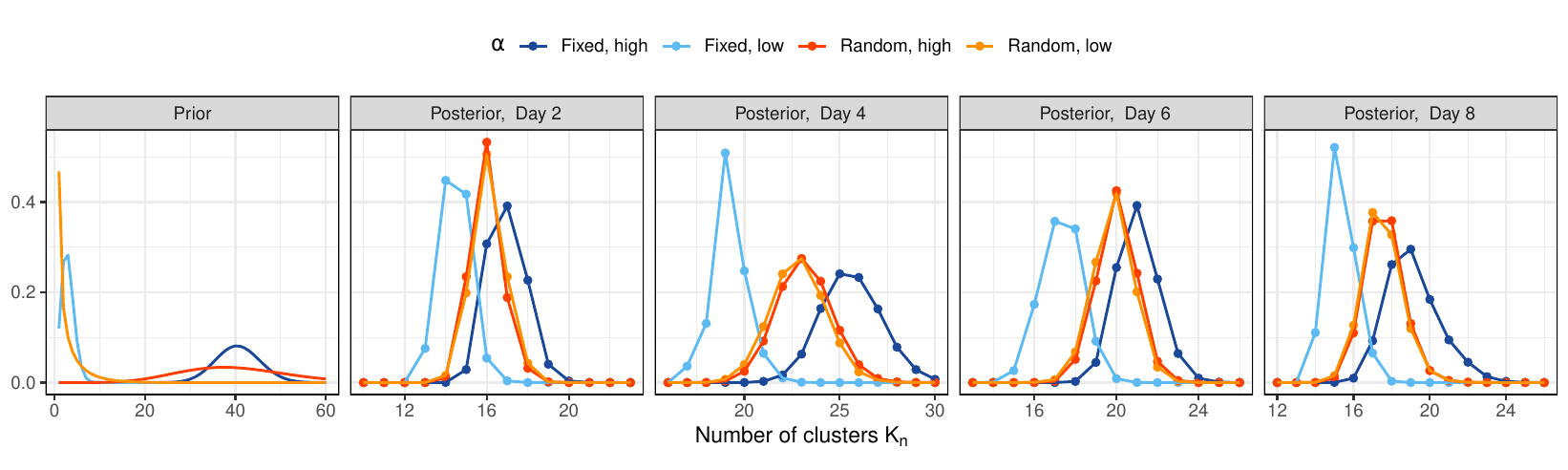}
\caption{Prior and posterior distribution of the number of clusters detected in the ant-to-ant binary interaction networks. Different colors refer to the four models tested. ``Fixed, low'' (light blue) refers to case (i) when $\alpha = 0.4$. ``Fixed, high'' (dark blue) is case (ii) when $\alpha = 18$. ``Random, low'' (light orange) is case (iii) with $\alpha \sim \mathrm{Sg}(0.6, 0.2, 149)$, and ``Random, high'' refer to case (iv) with $\alpha \sim \mathrm{Sg}(8, 0.2, 149)$.}
\label{fig:Prior_Post_Kn}
\end{figure}

\begin{figure}[h!]
\centering
\includegraphics[width = \linewidth]{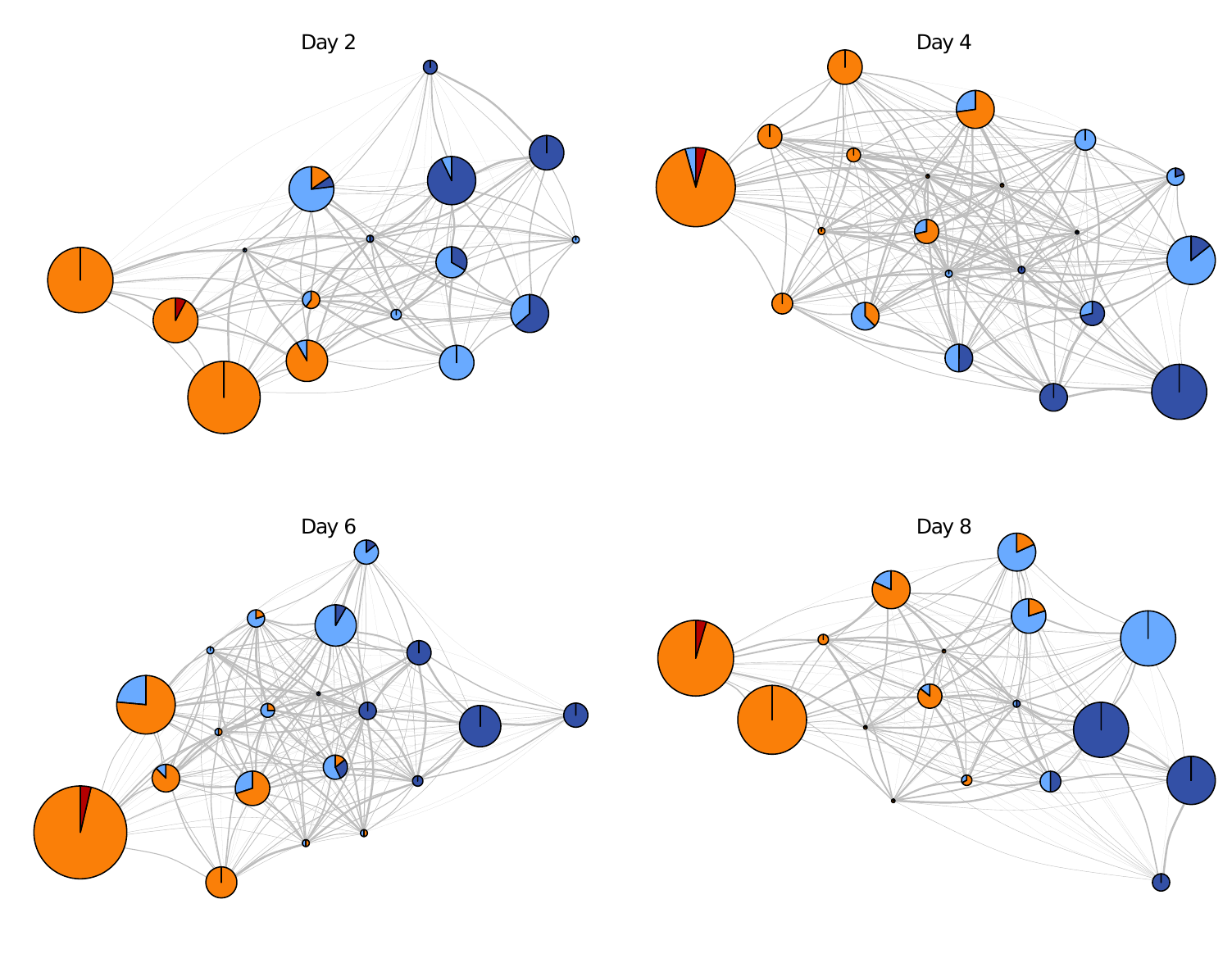}
\caption{Network representation of the inferred partition in the four networks displayed in Figure~\eqref{fig:ant_networks}. Nodes represent the retrieved clusters, with size determined by the number of ants they contain. Colors reflect the composition of each cluster according to the groups identified by \citet{Mersch2013}: foragers (dark blue), cleaners (light blue), and nurses (orange). The queen is indicated in red. We obtain the node positions through force-directed placement \citep{Fruchterman_Reingold_1991}. The width of the connections is determined by the posterior mean for the estimated block probabilities $\nu_{h, h', s}$, ignoring the ones below 0.1 for aesthetic reasons.}
\label{fig:Posteriror_networks}
\end{figure}

To further investigate the sub-communities retrieved by our model, we look at the posterior obtained from the Stirling-gamma process in model (iii). As we can see from Figure~\ref{fig:Prior_Post_Kn}, the average number of clusters detected in each day is much larger than the original grouping suggested by \citet{Mersch2013}. As such, the stochastic block model in equation~\eqref{eq:sbm} recovers a more complex ant organization than the one originally proposed, effectively detecting worker sub-communities. Figure~\ref{fig:Posteriror_networks} displays the posterior partitions obtained by minimizing the variation of information metric leveraging on the approach of \citet{Wade_Ghahramani_2018}. 
We observe 16, 21, 20,  and 17 clusters on days 2, 4, 6 and  8, respectively. Such differences are due to the within-day variability in worker interactions. These numbers are higher than the originally proposed three-cluster division in \citet{Mersch2013}, which is also apparent in the block structures in Figure~\ref{fig:ant_networks}. This is, however, in line with our model, which is designed to detect sub-communities within each day. In particular, we are able to isolate the groups uniquely characterized by nurses (in orange) and foragers (in dark blue), and identify the nurses and the forages that interact the most with the Cleaners (in light blue). Indeed, members of this last group tend to be co-clustered with the other two as they play a fundamental role in handling the passage of information within the colony. Finally, we are able to detect the sub-community of nurses that interacts the most with the queen. Such core structure of the social organization remains stable across days.

The model in equation~\eqref{eq:sbm}, endowed with a Dirichlet process prior, introduces a mild borrowing of information across the $N$ time layers through the shared precision $\alpha$.  There exists a wide variety of more complex models that account for repeated partitions of the same nodes where a stronger borrowing of information can be imposed. These include, for example, time-varying stochastic block models \citep{Matias_Miele_2016, Miele_Matias_2017} and multi-layer stochastic block models \citep{StanleyNatalie2016CNLw, Barbillon_2016}. While our analysis is meant to illustrate the robustness of having a random $\alpha$ in Dirichlet process mixtures, it will be interesting in the future to explore these approaches within a nonparametric framework using a Stirling-gamma process, possibly accounting for the many covariates available for the ants in each day.

\section{Discussion}\label{sec:discussion}

Our proposed Stirling-gamma prior was motivated by improving robustness to prior choice and transparency in prior elicitation in Dirichlet process mixture models. Fixing the precision parameter is a poor choice in most applications, since it implies a highly informative prior for the induced number of clusters. While the usual gamma prior can improve robustness to one's prior guess for the number of clusters, the implications of the gamma choice are unclear due to the lack of an analytically tractable form for the induced partition prior. The Stirling-gamma has the advantage of being more transparent in terms of the induced clustering prior than the gamma distribution while also marginally simplifying the posterior sampler thanks to its conjugacy and the efficient random variable generator introduced in this paper. In turn, the similarity between the two distributions can further aid interpretability to the gamma prior as well: when $\alpha\sim\mathrm{Ga}(a-b, b\log{m})$, then for sufficiently large $m$ we have that $\mathds{E}(K_m)$ is approximately $a/b$ from Theorem~\ref{theo:MarginalKm} and Proposition~\ref{pro:limit_gamma}.

The applications of the Stirling-gamma as a prior in mixture modeling are not limited to Dirichlet process mixtures alone. There exist many other Bayesian nonparametric models whose exchangeable partition probability function extends the Dirichlet process one. Such is the case, for instance, of the extended stochastic block model of \citet{Legramanti2022}, which modifies equation~\eqref{eq:dp_eppf} to incorporate cluster-specific covariates. As we show in the Supplementary material, the Stirling-gamma can be suitably used as a conjugate prior for $\alpha$ also under such a case, thus broadening its applicability to more complex models than the one considered in Section~\ref{sec:ants}. Another example is of the generalized mixture of finite mixtures of \citet{FruFru_2021}, where a heavy-tailed F-prior is employed over $\alpha$; refer again to the Supplementary material.

More broadly, the Stirling-gamma is of interest as a new heavy-tailed distribution having positive support. There are multiple other application areas in which this new distribution may be useful. For example, the Stirling-gamma could be used as the choice of exponential family distribution within a generalized linear model framework when the common log-normal or gamma choices lack sufficiently heavy tails for the data at hand. Alternatively, noting that the Dirichlet distribution arises by normalizing independent gamma random variables, one could obtain an alternative distribution on the probability simplex by normalizing Stirling-gamma random variables. This new distribution may be ideal at characterizing the case in which there are a small proportion of large probabilities with the remaining concentrated near zero; a common desirable behavior for shrinkage priors on the simplex. Other recent examples in this direction can be found in \citet{wang2024horseshoe}, who introduced a class of \emph{Pochhammer} distributions that include our Stirling-gamma as a special case. Interestingly, the Pochhammer prior is conjugate  to the Dirichlet-multinomial model and can suitably handle sparsity in compositional count data by inducing continuous shrinkage via a horseshoe-type behavior \citep{Carvalho_2010} thanks to the implied heavy tails. Finally, we could also leverage upon the property of infinite divisibility enjoyed by all generalized gamma convolutions \citep{Bondesson1979}, of which the Stirling-gamma is a member, to construct flexible finite-dimensional stochastic processes extending Dirichlet-multinomial partitions; see \citet{Lijoi_2023}.

One popular topic of debate in the recent literature is whether Dirichlet process mixtures are consistent in retrieving the ``true'' number of clusters \citep{Miller_2013, Zeng_Miller_Duan_2023, alamichel2023bayesian}. While keeping $\alpha$ fixed yields to an inconsistent posterior of $K_n$ in fairly broad settings \citep{JMLR:v15:miller14a}, randomization of the precision via a prior $\pi(\alpha)$ leads to theoretical consistency under very strong assumptions on the data generating process. In their main result, \cite{Ascolani_2022} assume complete separability of the true parameters identifying the clusters and consider only classes of kernels with bounded support; the standard normal mixture example of \citet{Miller_2013}, for instance, does not abide by these assumptions. Hence, while our Stirling-gamma prior does not meet the prior properties required by \cite{Ascolani_2022} for their proof (see the Supplementary material), additional theoretical investigation is still needed to gain a deeper understanding of the problem under more general settings. Based on the simulation in Section~\ref{sec:simulations}, we conjecture that consistency could be retrieved by letting $\pi(\alpha)$ converge to a point mass at zero as $n\to\infty$ at the appropriate rate, as happens with the conjugate Stirling-gamma prior; more work will be dedicated in the future to come to a conclusive statement. See also \cite{Ohn_Lin_2023} for properties of the posterior of Dirichlet process mixtures when $\alpha$ approaches zero deterministically as $n$ grows.

When the reference sample size $m$ of a Stirling-gamma prior $\alpha\sim \mathrm{Sg}(a, b, m)$ diverges, the total number of clusters generated from a Stirling-gamma process as $n\to\infty$ is negative binomial-distributed. Moreover, having $\alpha = \lambda/\log{m}$ with $m\to\infty$  makes $K_\infty - 1$ follows a Poisson distribution. This suggests a parallelism with mixtures of finite mixtures \citep{Richardson_green_1997, Nobile_2004, Miller_Harrison_2018}, which are a class of Gibbs-type processes \citep{Gnedin2005, DeBlasi2015} where the number of mixture components $k^*$ is itself a parameter of the model and follows a discrete prior of choosing, like $k^*-1 \sim\mathrm{NegBin}(r, p)$. Interestingly, these models have been shown to consistently retrieve the true number of clusters in general settings \citep{Miller_2023}. Based on our results, Dirichlet process mixtures can \emph{asymptotically} induce a finite number of clusters \emph{a priori} as well, provided that $\alpha\to0$ at a logarithmic rate. On the contrary, models such as generalized mixtures of finite mixtures introduced by \citet{FruFru_2021} converge to a Dirichlet process prior as $k^* \to \infty$. However, despite an apparent asymptotic equivalence, Dirichlet process mixtures always favor more imbalanced partitions than their mixture of finite mixture counterpart, since $\alpha$ does not influence the mechanism determining the cluster sizes in equation~\eqref{eq:dp_eppf}; refer to the insightful contribution of \citet{Changwoo_2022}. In future work, it will be interesting to study the implied asymptotic behavior for $K_n$ and the associated partition balancedness when randomizing the hyperparameters of Gibbs-type processes with infinite clusters with a prior that counterbalances the growth of $K_n$ itself. For example, one could design a prior for the strength hyperparameter $\sigma$ in a Pitman--Yor whose prior induces a polynomial decay at zero. Refer to \citet{DeBlasi2015} for details.

\section*{Acknowledgement}
This project has received funding from the European Research Council under the European Union’s Horizon 2020 research and innovation programme (grant agreement No 856506). The authors would like to extend their gratitude to Yuansi Chen, Antonio Lijoi, and Igor Pr\"{u}nster for their precious suggestions. We are also grateful to the associate editor and the three anonymous referees, who helped improving on the quality of the paper. 

\section*{Supplementary Material}
The Supplementary material includes the proofs of the statements in the paper, additional statements, details of the random number generator for the Stirling-gamma distribution, and a simulation study on data from simulated networks. Code to sample from the Stirling-gamma is available at \url{https://github.com/alessandrozito/ConjugateDP}.

\bibliographystyle{chicago}
\bibliography{references}

\newpage

\setcounter{page}{1}
\setcounter{section}{0}
\setcounter{table}{0}
\setcounter{figure}{0}
\setcounter{equation}{0}
\renewcommand{\theHsection}{SIsection.\arabic{section}}
\renewcommand{\theHtable}{SItable.\arabic{table}}
\renewcommand{\theHfigure}{SIfigure.\arabic{figure}}
\renewcommand{\thepage}{S\arabic{page}}  
\renewcommand{\thesection}{S\arabic{section}}   
\renewcommand{\theequation}{S\arabic{equation}}   
\renewcommand{\thetable}{S\arabic{table}}   
\renewcommand{\thefigure}{S\arabic{figure}}
\renewcommand{\thetheorem}{S\arabic{theorem}}
\renewcommand{\thelemma}{S\arabic{lemma}}
\renewcommand{\theproposition}{S\arabic{proposition}}
\renewcommand{\thecorollary}{S\arabic{corollary}}
\renewcommand{\theequation}{S\arabic{equation}}

\section*{Supplementary Material}
This document contains the proofs for all the statements of the main manuscript, additional simulations, and a rejection algorithm to sample from the Stirling-gamma distribution. The Supplement is divided as follows. Section~\ref{sec:lemmas} contains some preliminary lemmas that are useful for the main proofs. Section~\ref{sec:proofs} reports the proofs of the statements in the main manuscript. Section~\ref{sec:additional_res} presents additional useful theoretical results. Section~\ref{sec:proof_appendix} presents the proofs for the formulas of the normalizing constants.  Section~\ref{sec:FurtherResults} frames the Stirling-gamma within the consistency results of \citetSupp{Ascolani_2022}, and discusses some additional partition models where it can be useful. Section~\ref{sec:simul} presents an additional simulation study on the \emph{population of partition} framework, and reports the effective sample sizes for the simulation in Section~\ref{subsec:normalDPM}.  Throughout the document, we will write that $a(n) \sim b(n)$ as $n\to \infty$ to indicate that $\lim_{n\to\infty} a(n)/b(n) = 1$. While $\sim$  indicates ``is distributed as'' in the main document, here we employ such a slight abuse of notation to ease the readability of the statements below.

\section{Preliminary lemmas}\label{sec:lemmas}
This Section reports some preliminary lemmas that will be useful for the proofs presented in Section~\ref{sec:proofs}. We begin by recalling two asymptotic approximations for the gamma function:
\begin{align}
    \Gamma(m) \sim \sqrt{2\pi} e^{-m} m^{m - 1/2},  \qquad &m\to \infty; \label{eq:Gamma_Stirling_approx} \\
    \Gamma(z)  \sim 1/z, \qquad   &z\to 0. \label{eq:Gamma_zero_approx}
\end{align}
Equation~\eqref{eq:Gamma_Stirling_approx} is the famous Stirling approximation \citepSupp[][equation 6.1.37, p257]{abramowitz_stegun_1972}. Equation~\eqref{eq:Gamma_zero_approx} instead comes from taking the limit for $z\to 0$ to the product formula for the gamma function, since $\lim_{z\to0}\Gamma(z) = \lim_{z\to0}\Gamma(z+1)/z = \lim_{z\to0}1/z$. Then, the following lemmas hold.

\begin{lemma}\label{lem:exponential_limit}
For any $a \in \mathds{R}$, we have
$$\lim_{m\to\infty} \left(\frac{a}{\log{m}} + m\right)^{\frac{a}{\log{m}}} = e^a$$
\end{lemma}
\begin{proof}
By letting $x = \log{m}$ and collecting the term $e^x$, the limit simplifies as $$\lim_{x\to\infty} e^a (ae^{-x}/x + 1)^{a/x} = e^a.$$
\end{proof}

\begin{lemma}\label{lem:PoissonLaplace}
For any $x, z > 0$, we have 
$$
\lim_{m\to\infty} \frac{(xz/\log{m})_m}{(x/\log{m})_m} = ze^{zx - x}
$$
\end{lemma}
\begin{proof}
Recall that the ascending factorial can be defined as $(x)_a = \Gamma(x + a)/\Gamma(x)$. We can then rewrite the limit as 
$$
\lim_{m\to\infty} \frac{(xz/\log{m})_m}{(x/\log{m})_m} = \lim_{m\to \infty} \frac{\Gamma(x/\log{m})}{\Gamma(xz/\log{m})} \times \frac{\Gamma(xz/\log{m} + m)}{\Gamma(x/\log{m} + m)}.
$$
We study each fraction separately. By relying on the approximation in equation~\eqref{eq:Gamma_zero_approx}, we have  
\begin{equation}\label{eq:gamma_over_log}
\Gamma\left(\frac{a}{\log{m}}\right) \sim \frac{\log{m}}{a}, \qquad m \to \infty,
\end{equation}
for any $a>0$. Thus, the limit of the first fraction is equal to
$$
\lim_{m\to \infty} \frac{\Gamma(x/\log{m})}{\Gamma(xz/\log{m})} = \lim_{m\to \infty} \frac{\log{m}}{x} \frac{xz}{\log{m}} = z.
$$
From equation~\eqref{eq:Gamma_Stirling_approx}, we can also write that 
\begin{equation}\label{eq:gamma_over_log_m}
\Gamma\left(\frac{a}{\log{m}}  + m\right) \sim \sqrt{2\pi} e^{-\frac{a}{\log{m}}-m} \left(\frac{a}{\log{m}} + m\right)^{\frac{a}{\log{m}} + m - \frac12}, \qquad m \to \infty,
\end{equation}
for any $a>0$. But then, thanks to the result in Lemma~\ref{lem:exponential_limit}, the second fraction simplifies as
\begin{align*}
\lim_{m\to \infty} \frac{\Gamma(xz/\log{m} + m)}{\Gamma(x/\log{m} + m)} &= \lim_{m\to \infty} e^{\frac{-xz + x}{\log{m}}}\left(\frac{xz}{\log{m}} + m\right)^{\frac{xz}{\log{m}} + m - \frac{1}{2}}\left(\frac{x}{\log{m}} + m\right)^{-\frac{x}{\log{m}} - m + \frac{1}{2}} \\
&= \lim_{m\to \infty} \left(\frac{xz}{\log{m}} + m\right)^{\frac{xz}{\log{m}}}\left(\frac{x}{\log{m}} + m\right)^{-\frac{x}{\log{m}}} \left(\frac{xz + m\log{m}}{x + m\log{m}}\right)^{m - \frac{1}{2}} \\
&= \lim_{m\to \infty} \left(\frac{xz}{\log{m}} + m\right)^{\frac{xz}{\log{m}}}\left(\frac{x}{\log{m}} + m\right)^{-\frac{x}{\log{m}}}\\
&= e^{zx - x}
\end{align*}
This completes the proof.
\end{proof}

Lemma~\ref{lem:PoissonLaplace} will be useful when proving the convergence of $K_m$ to the negative binomial distribution under the Stirling-gamma process. The next two lemmas characterize the asymptotic behavior of the normalizing constant of the Stirling-gamma distribution.

\begin{lemma}\label{lem:LogEppfApprox}
The following asymptotic approximation holds for any $x > 0$:
$$
\frac{1}{(\log{m})^a}\frac{x^{a-1}}{\{(x/\log{m})_m\}^b} \sim g(m,a, b) x^{a-b -1} e^{-bx}, \qquad m\to \infty,
$$
where $g(m,a, b) = (2\pi)^{-b/2}(\log{m})^{b-a} e^{bm} m^{-bm + b/2}$.
\end{lemma}
\begin{proof}
This asymptotic behavior follows from equations~\eqref{eq:gamma_over_log} and~\eqref{eq:gamma_over_log_m}. In particular, 
by expressing the ascending factorial in the denominator as a ratio of gamma functions, we have 
\begin{align*}
    \lim_{m\to \infty} &\frac{1}{(\log{m})^a}\frac{x^{a-1}}{\{(x/\log{m})_m\}^b}  \\
    &=\lim_{m\to\infty}  (2\pi)^{-\frac{b}{2}} (\log{m})^{b-a} \ e^{\frac{bx}{\log{m}} + bm} \ \left(\frac{x}{\log{m}} + m\right)^{-\frac{x}{\log{m}} - bm + \frac{b}{2}} x^{a-b-1}\\
    &=\lim_{m \to \infty} (2\pi)^{-\frac{b}{2}}  (\log{m})^{b-a} \ e^{bm} m^{-bm +\frac{b}{2}}  \ x^{a-b-1} e^{-bx},
\end{align*}
where the simplifications follow from $bx/\log{m} \to 0$ and the limit in Lemma~\ref{lem:exponential_limit}.  We complete the proof by calling $g(m,a, b) = (2\pi)^{-b/2}(\log{m})^{b-a} e^{bm} m^{-bm + b/2}$ the part that  depends on $m$, $a$ and $b$ and not on $x$.
\end{proof}

\begin{lemma}\label{lem:NormConst_asymp}
When $a, b >0$ and $1<a/b<m$, the following limit holds for the normalizing constant of a Stirling-gamma distribution:
$$
\lim_{m \to \infty} \frac{g(m,a,b)}{\mathcal{S}_{a, b, m}} = \frac{b ^ {a - b}}{\Gamma(a - b)},
$$
where $g(m, a, b)$ is defined in Lemma~\ref{lem:LogEppfApprox}.
\end{lemma}
\begin{proof}
The proof is a direct consequence of Lemma~\ref{lem:LogEppfApprox} and of the monotone convergence theorem. Consider the change of variable $\alpha = x/\log{m}$. Then, the normalizing constant can be rewritten as
$$
\mathcal{S}_{a,b,m} = \int_{\mathds{R}_+} \frac{\alpha^{a-1}}{\{(\alpha)_m\}^b}\textrm{d}\alpha = \int_{\mathds{R}_+} \frac{1}{(\log{m})^a}\frac{x^{a-1}}{\{(x/\log{m})_m\}^b} \text{d}x.
$$
Provably, both integrands are monotonically decreasing in $m$. By monotone convergence theorem, this ensures that the limit and the integral can be interchanged. Invoking the approximation of Lemma~\ref{lem:LogEppfApprox}, we have
\begin{align*}
    \lim_{m\to\infty} \frac{g(m,a, b)}{\mathcal{S}_{a,b,m}} = \lim_{m\to\infty} \frac{g(m, a, b)}{g(m,a, b)\int_{\mathds{R_+}} x^{a-b-1}e^{-bx}\text{d}x} = \frac{b^{a-b}}{\Gamma(a-b)}.
\end{align*}
Notice that $\mathcal{S}_{a,b, m} <\infty$ when $1<a/b<m$, as we show in Proposition~\ref{pro:NormConst_finite} below. 
\end{proof}

We conclude the section by providing an expression for the Laplace transform of the distribution for $K_m$ conditional on $\alpha$. 
\begin{lemma}\label{lem:laplace}
Let $\theta_1, \ldots, \theta_m$ be a sample from a Dirichlet process with precision parameter $\alpha$.
The conditional Laplace transform of the number of clusters $K_m$ is 
$$\mathds{E}(e^{-tK_m}\mid \alpha) = \frac{(\alpha e^{-t})_m}{(\alpha)_m}, \qquad t \geq 0.$$
\end{lemma}
\begin{proof}
The result follows directly from the relationship between the ascending factorial and the signless Stirling numbers of the first kind detailed in \citetSupp{Charalambides_2005}:
\begin{align*}
\mathds{E}(e^{-tK_m}\mid \alpha) = \sum_{k = 1}^m e^{-tk} \frac{\alpha^k}{(\alpha)_m}|s(m, k)| = \frac{1}{(\alpha)_m}\sum_{k = 1}^n (\alpha e^{-t})^k|s(m, k)| = \frac{(\alpha e^{-t})_m}{(\alpha)_m},
\end{align*}
for any $t > 0.$
\end{proof}
\section{Main proofs}\label{sec:proofs}
\subsection{Proof of Proposition~\ref{pro:limit_gamma}}
\begin{proof}
To prove the convergence in distribution, it is sufficient to show that the Laplace transform of the quantity $\alpha \log{m}$ converges to that of a gamma distribution. In particular, for any $t >0$ we have
$$
\mathds{E}(e^{-t \alpha \log{m}}) = \int_{\mathds{R}_+} \frac{1}{\Vabm}e^{-t \alpha \log{m}} \frac{\alpha^{a-1}}{\{(\alpha)_m\}^b}\textrm{d}\alpha = \int_{\mathds{R}_+} \frac{1}{\Vabm}\frac{e^{-t x}}{(\log{m})^a} \frac{x^{a-1}}{\{(x/\log{m})_m\}^{b}}\textrm{d}x,
$$
where the second equality comes from changing the integration variable to $x = \alpha \log{m}$.
By relying on the bounded convergence Theorem, we can interchange the integral and the limit when $m\to\infty$. Thus, from Lemma~\ref{lem:LogEppfApprox} and \ref{lem:NormConst_asymp}, we can write that
\begin{align*}
    \lim_{m\to\infty} \mathds{E}(e^{-t\alpha\log{m}}) &= \lim_{m\to\infty} \frac{g(m,a, b)}{\Vabm} \int_{\mathds{R}_+} x^{a-b-1} e^{- (b+t)x}\text{d}x \\
    &= \frac{b^{a-b}}{\Gamma(a-b)}\int_{\mathds{R}_+} x^{a-b-1} e^{- (b+t)x}\text{d}x  = \left(\frac{b}{b+t}\right)^{a-b},
\end{align*}
which is the Laplace transform of a $\mathrm{Ga}(a-b, b)$ random variable.
\end{proof}

\subsection{Proof of Proposition~\ref{pro:moments}}
\begin{proof}
By definition of expected value, we have 
$$
\mathds{E}(\alpha^s) = \frac{1}{\mathcal{S}_{a,b,m}} \int_{\mathds{R}_+} \frac{\alpha^{a + s - 1}}{\{(\alpha)_m\}^b}\text{d}\alpha = \frac{\mathcal{S}_{a+s,b,m}}{\mathcal{S}_{a,b,m}},
$$
where the integral at the numerator is finite if and only if $0 <s < mb-a$. See  Proposition~\ref{pro:NormConst_finite} for a proof.
\end{proof}

\subsection{Proof of Theorem~\ref{theo:EPPF_SGP}}
\begin{proof}
Let $\alpha \sim \mathrm{Sg}(a, b,m)$ in equation~(1). Then, 
\begin{align*}
    \mathds{P}(\Pi_n = \{C_1, \ldots, C_k\}) = \frac{1}{\mathcal{S}_{a,b,m}} \left\{\int_{\mathds{R}_+} \frac{\alpha^{a + k - 1}}{\{(\alpha)_m\}^b(\alpha)_n}\text{d}\alpha\right\}\prod_{j = 1}^{k} (n_j - 1)!
\end{align*}
Calling $\mathscr{V}_{a, b, m}(n, k) =\int_{\mathds{R}_+} \alpha^{a + k - 1}/[\{(\alpha)_m\}^b(\alpha)_n]^{-1}\text{d}\alpha$ and $\mathscr{V}_{a, b, m}(1, 1) =\mathcal{S}_{a,b,m}$ completes the proof.
\end{proof}

\subsection{Proof of Theorem~\ref{theo:MarginalKm}}
\begin{proof}
We begin by showing the shape of the probability mass function for $K_m$. This follows directly from the formula in \citetSupp{antoniak1974}. In particular, 
\begin{align*}
\mathds{P}(K_m = k) &= \int_{\mathds{R}_+} \mathds{P}(K_m = k\mid \alpha) \pi(\alpha)\text{d}\alpha = \frac{1}{\Vabm}\left[\int_{\mathds{R}_+} \frac{\alpha^{a + k - 1}}{\{(\alpha)_m\}^b(\alpha)_m}\text{d}\alpha\right]|s(m, k)| \\
&= \frac{\mathscr{V}_{a, b, m}(m, k)}{\mathscr{V}_{a, b, m}(1, 1)}|s(m, k)|,
\end{align*}
for any $k = 1, \ldots, m$. We now provide a formula for the mean and the variance of the distribution above. First, recall that the expected value and the variance of $K_m$ conditional on $\alpha$ are
\begin{align}
\mathds{E}(K_m\mid \alpha) &=  \alpha\{\psi(\alpha + m)- \psi(\alpha)\} \label{eq:EKn_alpha} \\
\mathrm{var}(K_m \mid \alpha) &=  \alpha\{\psi(\alpha + m)- \psi(\alpha)\} +  \alpha^2\{\psi'(\alpha + m)- \psi'(\alpha)\} \label{eq:varKn_alpha},
\end{align}
where $\psi(x) = \Gamma'(x)/\Gamma(x)$ is the digamma function and $\psi'(x)$ is its derivative, called trigamma function. The property that $\mathds{E}(K_m) = a/b$ follows immediately from \citetSupp{Diaconis_Yilvisaker_1979}: writing $\eta = \log{\alpha}$, we have $$\mathds{P}(K_m = k \mid \eta) \propto \exp\{k\eta  - \mathcal{K}(\eta,m)\},$$ with $\mathcal{K}(\eta,m) = \log \Gamma(e^\eta + m) - \log\Gamma(e^\eta)$. Thus, the associated conjugate prior in the natural parametrization is $p(\eta) \propto \exp\{\tau_0 k_0 \eta - \tau_0 \mathcal{K}(\eta,m)\}$. Moreover, we have that $\text{d} \mathcal{K}(\eta,n)/\text{d}\eta = e^\eta\{\psi(e^\eta + n) - \psi(e^\eta)\}$.
From Theorem 2 in~\citetSupp{Diaconis_Yilvisaker_1979}, we have
$$
\mathds{E}(K_m) = \mathds{E}\{\mathds{E}(K_m \mid \eta)\} =  E\left\{\frac{\text{d}}{\text{d}\eta}\mathcal{K}(\eta,m)\right\}= E[e^\eta\{\psi(e^\eta + m) - \psi(e^\eta)\}] = k_0.
$$
Substituting again $\alpha = e^\eta$ and calling $k_0 = a/b$ and $\tau_0 = b$ in the above proves the statement.

To prove the expression for the variance, we rely on the law of iterated variances, that is 
$$\mathrm{var}(K_m) = \mathds{E}\{\mathrm{var}(K_m\mid \alpha)\}+ \mathrm{var}\{\mathds{E}(K_m\mid \alpha)\}.$$
Both terms can be expressed as a function of the quantity 
$\mathcal{D}_{a, b, m} = E[\alpha^2\{\psi'(\alpha) - \psi'(\alpha +m)\}] = \mathds{E}\{\sum_{i = 0}^{m - 1}\alpha^2/(\alpha + i)\}$. 
To simplify the expression, we rely on integration by parts. Consider the following functions:
$$I(\alpha) = \frac{\alpha^{a + 1}}{\{(\alpha)_m\}^b},\qquad M(\alpha) = \psi(\alpha + m) - \psi(\alpha),$$
whose derivatives with respect to $\alpha$ are equal to
$$
I'(\alpha) = \frac{(a+1)}{\alpha}I(\alpha) - bM(\alpha)I(\alpha), \qquad M'(\alpha) = \psi'(\alpha + m) - \psi'(\alpha).
$$
Then, the integral simplifies as 
\begin{align*}
\mathcal{D}_{a, b, m} &= -\frac{1}{\mathcal{S}_{a, b, m}}\int_{\mathds{R}_+} M'(\alpha)I(\alpha)\text{d}\alpha \\
&= -\frac{1}{\mathcal{S}_{a, b, m}}\left|M(\alpha)I(\alpha)\right|_{\alpha = 0}^{\infty} - \frac{b}{\mathcal{S}_{a, b, m}} \int_{\mathds{R}_+} M(\alpha)^2I(\alpha)\text{d}\alpha + \frac{a+1}{\mathcal{S}_{a, b, m}}\int_{\mathds{R}_+} \frac{M(\alpha)}{\alpha}I(\alpha)\text{d}\alpha \\
& = -b \mathds{E}\{\alpha^2M(\alpha)^2\} + (a+1)\mathds{E}\{\alpha M(\alpha)\}  \\
& = - b \mathds{E}\{\alpha^2M(\alpha)^2\} + \frac{a(a+1)}{b} \\
&= - b\left[ \mathds{E}\{\alpha^2M(\alpha)^2\}   - \frac{a^2}{b^2}\right] + \frac{a}{b}.
\end{align*}
%
%
Moreover, from equation~\eqref{eq:EKn_alpha}, we can write 
\begin{align*}
    \mathrm{var}\{\mathds{E}(K_m\mid \alpha)\} &= \mathrm{var}\{\alpha M(\alpha)\} = \mathds{E}\{\alpha^2M(\alpha)^2\} - [\mathds{E}\{\alpha M(\alpha)\}]^2 \\
    &= \mathds{E}\{\alpha^2M(\alpha)^2\} - \frac{a^2}{b^2} \\
    &= \frac{a}{b^2} - \frac{\mathcal{D}_{a, b, m}}{b},
\end{align*}
where the last equality comes from plugging in the new expression for $\mathcal{D}_{a, b, m}$. Finally, from equation~\eqref{eq:varKn_alpha}, we also have
\begin{align*}
    \mathds{E}\{\mathrm{var}(K_m\mid \alpha)\} = \mathds{E}\{\alpha M(\alpha)\} + \mathds{E}\{\alpha^2 M'(\alpha)\} = \frac{a}{b} - \mathcal{D}_{a,b,m}.
\end{align*}
Combining the last two equalities in the law of iterated variance proves the result. 
\end{proof}

\subsection{Proof of Theorem~\ref{theo:NegBin}}
\begin{proof}
The proof of convergence in distribution to the negative binomial relies on Lemmas~\ref{lem:PoissonLaplace}, \ref{lem:NormConst_asymp} and \ref{lem:laplace}. Substituting $\alpha = x/\log{m}$ in the integral, the marginal Laplace transform of $K_m$ is
\begin{align*}
\mathds{E}(e^{-t K_m}) &= \mathds{E}\{\mathds{E}(e^{-t K_m}\mid \alpha)\} = \frac{1}{\mathcal{S}_{a, b, m}} \int_{\mathds{R}_+}\frac{(\alpha e^{-t})_m}{(\alpha)_m} \frac{\alpha^{a-1}}{\{(\alpha)_m\}^b}\text{d}\alpha  \\
&= \frac{1}{\mathcal{S}_{a, b, m}} \int_{\mathds{R}_+}\frac{(x e^{-t}/\log{m})_m}{(x/\log{m})_m}\frac{1}{(\log{m})^a} \frac{x^{a-1}}{\{(x/\log{m})_m\}^b}\text{d}x
\end{align*}
Then, the limit is
\begin{align*}
    \lim_{m\to\infty} \mathds{E}(e^{-t K_m}) &= \lim_{m\to\infty} \frac{1}{\mathcal{S}_{a, b, m}} \int_{\mathds{R}_+}\frac{(x e^{-t}/\log{m})_m}{(x/\log{m})_m}\frac{1}{(\log{m})^a} \frac{x^{a-1}}{\{(x/\log{m})_m\}^b}\text{d}x \\
     &= \lim_{m\to\infty} \frac{g(a,b,m)}{\mathcal{S}_{a,b,m}}\int_{\mathds{R}_+} e^{e^{-t}x - x -t} x^{a-b-1}e^{-bx}\text{d}x\\
     &= \frac{e^{-t}b^{a-b}}{\Gamma(a-b)}\int_{\mathds{R}_+} x^{a-b-1}e^{-(1+ b - e^{-t})x} \text{d}x \\
     &= e^{-t}\left(\frac{b}{1+b - e^{-t}}\right)^{a-b},
\end{align*}
which is the Laplace transform of $1 + \mathrm{Negbin}(a-b, b/(b + 1) )$, whose probability mass function is 
$$
\mathds{P}(K_\infty = k) = \frac{\Gamma(a-b + k -1)}{(k-1)!\Gamma(a-b)} \left(\frac{1}{b+1}\right)^{k-1}\left(\frac{b}{b+1}\right)^{a-b}, \quad k = 1, 2, \ldots
$$
The limit and integral can be interchanged thanks to the bounded convergence theorem, since the Laplace transform is always bounded by 1. 
\end{proof}

\subsection{Proof of Proposition~\ref{pro:poisson}}
\begin{proof}
The proof follows directly from Lemma~\ref{lem:PoissonLaplace} and Lemma~\ref{lem:laplace}. In particular, substituting $\alpha = \lambda/\log{m}$ in the limit, we have
$$
\lim_{m\to\infty} \mathds{E}(e^{-t K_m}) = \lim_{m\to\infty}\frac{(\alpha e^{-t})_m}{(\alpha)_m} = \lim_{m\to\infty}\frac{(\lambda e^{-t}/\log{m})_m}{(\lambda/\log{m})_m} = e^{e^{-t}\lambda - \lambda - t},
$$
which is the Laplace transform of $1 + \text{Po}(\lambda)$.
\end{proof}

\subsection{Proof of Proposition~\ref{pro:conjugacy}}
\begin{proof}
The statement follows trivially from Bayes theorem:
$$
\pi(\alpha\mid \Pi_{n} = \{C_1, \ldots, C_k\}) \propto \pi(\alpha)\ \mathds{P}(\Pi_{n} = \{C_1, \ldots, C_k\} \mid \alpha) \propto \frac{\alpha^{a-1}}{\{(\alpha)_n\}^b}\frac{\alpha^{k}}{(\alpha)_n}.
$$
Therefore, $(\alpha\mid \Pi_{n} = \{C_1, \ldots, C_k\}) \sim \mathrm{Sg}(a + k, b + 1, n)$. Notice that if $a, b>0$ and $1< a/b < n$, then also $1 < (a + k)/(b+1) < n$, since $1 \leq k \leq n$. Thus, a proper prior implies automatically a proper posterior.
\end{proof}

\subsection{Proof of Theorem~\ref{theo:Conjugacy}}
\begin{proof}
The proof is similar to the one of Proposition 4. In particular, we have that 
\begin{align*}
\pi(\alpha\mid {\Pi}_{n,1}, \ldots, \Pi_{n, N}) &\propto \pi(\alpha)\prod_{s = 1}^N \mathds{P}(\Pi_{n, s} = \{C_{1, s}, \ldots, C_{k_s, s}\} \mid \alpha) \propto \frac{\alpha^{a-1}}{\{(\alpha)_n\}^b}\frac{\alpha^{\sum_{s = 1}^N k_s}}{\{(\alpha)_n\}^N},
\end{align*}
which proves the statement.
\end{proof}

\subsection{Proof of Proposition~\ref{pro:PosteriorAlpha}}
\begin{proof}
The proof naturally follows from equation~(2.10) in \citetSupp{Diaconis_Yilvisaker_1979}. Alternatively, we can derive the same result via Theorem 4, integrating over $\alpha\sim \mathrm{Sg}(a + N\bar{k}, b + N, n)$.
\end{proof}

\subsection{Proof of Proposition~\ref{pro:BetaPrime_ratio}}
\begin{proof}
The density of the generalized beta prime distribution  $\alpha \sim \mathrm{BeP}(a_0, b_0, r)$ is
$$
\pi_{\mathrm{BeP}}(\alpha) = \frac{(\alpha/r)^{a_0 - 1}(1 + \alpha/r)^{-a_0 - b_0}}{r\beta(a_0, b_0)}
$$
with $\alpha > 0$ and $\beta(a_0, b_0) = \Gamma(a_0)\Gamma(b_0)/\Gamma(a_0 + b_0)$ denoting the Beta function. Let $r(\alpha)$ be the function of $\alpha\geq 0$ defined as
$$
r(\alpha) = \left\{\prod_{i = 1}^{m-1}(\alpha + i)\right\}^{1/(m-1)} - \alpha.
$$
It is easy to see that $r(\alpha)$ is a monotonically increasing function of $\alpha$ whose minimum value is $r = r(0) = \Gamma(m)^{1/(m-1)}$. Since $(\alpha + r(\alpha))^{m-1} =  \prod_{i = 1}^{m-1}(\alpha + i)$, then
$$
(\alpha + r)^{m-1} \leq \prod_{i = 1}^{m-1}(\alpha + i), \qquad r = \Gamma(m)^{1/(m-1)},
$$
for every $\alpha >0$. Then, we always have that
$$
\frac{\alpha^{a-b-1}}{\prod_{i = 1}^{m-1} (\alpha + i)^b} \leq \frac{\alpha^{a-b-1}}{(\alpha + r)^{b(m-1)}} = r^{a - mb - 1} \frac{(\alpha/r)^{a-b-1}}{(1 + \alpha/r)^{b(m-1)}}.
$$
Hence, when $\alpha \sim \mathrm{BeP}(a-b, mb -a, r)$, with $r = \Gamma(m)^{1/(m-1)}$, we can write
\begin{align*}
\frac{\pi_\mathrm{Sg}(\alpha)}{\pi_\mathrm{BeP}(\alpha)} &= \frac{r\beta(a-b, mb-a)}{\mathcal{S}_{a, b,m}} \frac{\alpha^{a-b-1}}{\prod_{i = 1}^{m-1}(\alpha+i)^b} \frac{(1 + \alpha/r)^{b(m-1)}}{(\alpha/r)^{a-b-1}}  \\
& =\frac{(\alpha + r)^{b(m-1)}}{(\alpha+1)^{b}_{m-1}} \frac{\beta(a-b, mb-a)}{r^{mb-a}\mathcal{S}_{a,b,m}} \\
&\leq \frac{\beta(a-b, mb-a)}{r^{mb-a}\mathcal{S}_{a,b,m}}  = M < \infty,
\end{align*}
where the inequality follows after writing $\prod_{i =1}^{m-1}(\alpha + i) = (\alpha + 1)_{m-1}$. 
\end{proof}

\section{Additional results}\label{sec:additional_res}
\subsection{Finiteness of the normalizing constant}
\begin{proposition}\label{pro:NormConst_finite}
The normalizing constant of a Stirling-Gamma distribution, namely
$$\mathcal{S}_{a, b, m} = \int_{\mathds{R}_+}\frac{\alpha^{a-1}}{\{(\alpha)_m\}^b}\mathrm{d}\alpha,$$
is finite if and only if $a, b > 0$ and $m \in \mathds{N}$ with $1< a/b < m$.
\end{proposition}
\begin{proof}
Let $\mathcal{S}_{a, b, m} = \mathcal{S}_{a, b, m}^{0,1} + \mathcal{S}_{a, b, m}^{1, \infty}$, with $\mathcal{S}_{a, b, m}^{\ell,u}=\int_{\ell}^{u} \alpha^{a-1}\{(\alpha)_m\}^{-b}\mathrm{d}\alpha$. To prove the statement, it is sufficient to show that both integrals are finite. For $\mathcal{S}_{a, b, m}^{0,1}$, recall that for $0<\alpha<1$ and $b>0$, we have $\{(\alpha)_m\}^b = \alpha^b\{\prod_{j=1}^{m-1}(\alpha+j)\}^b > \alpha^b$, since $\prod_{j=1}^{m-1}(\alpha+j) \geq (m-1)! \geq 1$. This implies that $\mathcal{S}_{a, b, m}^{0,1} \leq \int_0^1 \alpha^{a-b-1}\text{d}\alpha$, which is finite if and only if $a/b >1$. For $\mathcal{S}_{a, b, m}^{1, \infty}$, instead, recall that $\lim_{\alpha\to\infty}\alpha^m\Gamma(\alpha)/\Gamma(\alpha+m)=1$ for any $m\geq 1$. But then, letting $q_1(\alpha) = \alpha^{a-1}\{(\alpha)_m\}^{-b}$ and $q_2(\alpha) = \alpha^{a-1-mb}$, we have that 
$\lim_{\alpha\to\infty} q_1(x)/q_2(x) =\lim_{\alpha\to\infty} \{\alpha^m\Gamma(\alpha)/\Gamma(\alpha+m)\}^b=1$, which implies that $\mathcal{S}_{a, b, m}^{1, \infty}=\int_1^\infty q_1(\alpha)\text{d}\alpha<\infty$  if and only if $\int_1^\infty q_2(\alpha)\text{d}\alpha<\infty$  by the limit comparison test. The latter integral is finite if and only if $a/b < m$ and $b>0$. Both sides require $b>0$, which in turn implies that also $a>0$. This completes the proof.
\end{proof}

\subsection{Theorem: convergence to a negative binomial via gamma prior}
\begin{theorem}
In the same setting of Theorem 3, let $\alpha \sim \mathrm{Ga}(a-b, b\log{m})$. Then, the following convergence in distribution holds:
$$
K_m \to  K_\infty, \quad K_\infty \sim 1 + \mathrm{Negbin}\left( a - b,\frac{b}{b + 1}\right), \quad m \to \infty. 
$$
\end{theorem}
\begin{proof}
The proof follows similarly to the one of Theorem 3 in the previous section, substituting again $\alpha = x/\log{m}$ in the integral:
\begin{align*}
    \lim_{n\to\infty} \mathds{E}(e^{-t K_m}) 
    &= \lim_{m\to\infty} \int_{\mathds{R}_+}\frac{(b\log{m})^{a-b}}{\Gamma(a-b)}\frac{(\alpha e^{-t})_m}{(\alpha)_m} \alpha^{a-b-1}e^{-\alpha b \log{m}} \text{d}\alpha \\
     &=\lim_{m\to\infty} \frac{b^{a-b}}{\Gamma(a-b)}\int_{\mathds{R}_+}\frac{(x e^{-t}/\log{m})_m}{(x/\log{m})_m} x^{a-b-1}e^{-bx} \text{d}x \\
     &= \frac{e^{-t}b^{a-b}}{\Gamma(a-b)}\int_{\mathds{R}_+} x^{a-b-1}e^{-(1+ b - e^{-t})x} \text{d}x = e^{-t}\left(\frac{b}{1+b - e^{-t}}\right)^{a-b},
\end{align*}
which is again the Laplace transform of $1 + \mathrm{Negbin}(a-b, b/(b + 1))$.
\end{proof}

\subsection{Proposition: tail of the Stirling-gamma distribution}

\begin{proposition}\label{pro:Heavytails}
The Stirling-gamma distribution $\alpha\sim\mathrm{Sg}(a,b,m)$  is heavy-tailed, namely
$$
\lim_{x \to\infty} e^{tx}\mathds{P}(\alpha > x) = \infty,\quad t> 0. 
$$
\end{proposition}
\begin{proof}
The proof relies on the Stirling approximation of the Gamma function applied to the ascending factorial. Following equation~\eqref{eq:Gamma_Stirling_approx}, we have
\begin{align*}
\frac{1}{(x)_m} = \frac{\Gamma(x)}{\Gamma(x + m)} \sim  \frac{e^{-x}x^{x - 1/2}}{e^{-x - m} (x + m)^{x + m - 1/2}} \sim \frac{1}{(x + m)^m}, \qquad x \to \infty,
\end{align*}
since $\lim_{x \to \infty}\{x/(x + m)\}^{x - 1/2} = e^{-m}$. This implies that 
\begin{equation*}
\frac{x^{a-1}}{\{(x)_m\}^b} \sim \frac{x^{a-1}}{(x + m)^{mb}} \sim \frac{1}{x^{mb-a+1}}, \qquad x \to \infty,
\end{equation*}
because $mb>a$ by definition. But then, we write
\begin{align*}
\lim_{x\to\infty} e^{tx}\mathds{P}(\alpha>x) &= \lim_{x\to\infty} \frac{1}{\mathcal{S}_{a, b, m}}\frac{\int_{x}^{\infty} \alpha^{a-1}\{(\alpha)_m\}^{-b}\mathrm{d}\alpha}{e^{-tx}} = \lim_{x\to\infty} \frac{1}{\mathcal{S}_{a, b, m}} \frac{x^{a-1}\{(x)_m\}^{-b}}{te^{-tx}} \\
&= \lim_{x\to\infty} \frac{1}{\mathcal{S}_{a, b, m}} \frac{e^{tx}}{tx^{mb-a+1}} = \infty,
\end{align*}
where the second equality follows by looking at the ratio of the derivatives with respect to $x$ using L'H\^{o}pital's rule. 
\end{proof}

\begin{corollary}\label{cor:Gammatails}
Let $\pi_{\mathrm{Sg}}(\alpha)$ denote the density of $\alpha \sim \mathrm{Sg}(a,b,m)$, and $\pi_{\mathrm{Ga}}(\alpha)$ the density of $\alpha \sim \mathrm{Ga}(a- b, b\log{m})$. The following limit holds:
$$
\lim_{\alpha \to \infty} \frac{\pi_{\mathrm{Sg}}(\alpha)}{\pi_{\mathrm{Ga}}(\alpha)} = \infty.
$$
Hence, a Stirling-gamma has a heavier right tail than the gamma distribution. 
\end{corollary}
\begin{proof}
Let $\mathcal{C} = \Gamma(a - b)/\{\mathcal{S}_{a, b, m} (b\log{m})^{a-b}\}$ denote the ratio of the normalizing constants. Then, the limit of the ratio of the densities is equal to
\begin{align*}
\lim_{\alpha \to \infty} \frac{\pi_{\mathrm{Sg}}(\alpha)}{\pi_{\mathrm{Ga}}(\alpha)} = \lim_{\alpha \to \infty} \mathcal{C} \frac{\alpha^{a-1}\{(\alpha)_m\}^{-b}}{\alpha^{a-b-1}e^{-\alpha b \log{m}}} = \lim_{\alpha \to \infty} \mathcal{C} \frac{\alpha^b e^{\alpha b \log{m}}}{(\alpha + m)^{bm}} = \infty.
\end{align*}
This implies that the tail of the gamma distribution decays faster than that of the Stirling-gamma.
\end{proof}

\section{Proofs of the analytic expressions for the normalizing constants}\label{sec:proof_appendix}
\subsection{Prior closed form expressions}
We hereby show how the coefficients $\mathcal{S}_{a, b, m}$ introduced in Definition~\ref{def:StirlingGammaPdf} admit an explicit form. These depend on complete exponential Bell polynomials, which are defined as follows. Given the variables  $x_1, \ldots, x_s$ for $s\geq 1$, the $s\textrm{th}$ complete exponential Bell polynomial is 
\begin{equation}\label{eq:bell}
B_s(x_1, \ldots, x_s) = \sum_{(i_1, \ldots, i_s)\in I_s}\frac{s!}{i_1!i_2!\cdots i_s!} \left(\frac{x_1}{1!}\right)^{i_1}\left(\frac{x_2}{2!}\right)^{i_2}\cdots \left(\frac{x_s}{s!}\right)^{i_s},
\end{equation}
where $I_s$ is the set of all non-negative integers $\{i_1, \ldots, i_s\}$ that satisfy the equality constraint $i_1 + 2i_2+\ldots+ si_s = s$ \citepSupp{Charalambides_2005}. Recalling that $\mathscr{S}_{b, j}(x_1, \ldots, x_b)$ are defined in equation~\eqref{eq:Pfunction2}, the following holds for the prior normalizing constant. 

\begin{theorem}\label{theo:NormConst}
If $a, b \in \mathds{N}$, then 
\begin{equation*}\label{eq:NormConst}
\mathcal{S}_{a, b, m} =  \sum_{j=1}^{m-1}(-1)^{\bar{c}+bj}\frac{ j ^{\bar{c}}}{\{\Gamma(j) \Gamma(m-j)\}^b} \mathscr{S}_{b, j}(h_{j, 1}, \ldots, h_{j, b}),
\end{equation*}
where $\bar{c} = a - b - 1$, $\mathscr{S}_{b, j}$ is defined in equation~\eqref{eq:Pfunction2} and 
$$h_{j, s} = -(a-1)\frac{(s-1)!}{j^{s}} - b(s-1)!(H_{m-j-1,s} - H_{j,s}),$$
with $H_{j,s} = \sum_{i = 1}^j 1/i^s$ being the $j\textrm{th}$ generalized harmonic number of order $s$.
\end{theorem}
Notice that the expression above can be simplified when $b=1$, as in the following Corollary. 
\begin{corollary}\label{cor:NormConst_b1}
The normalizing constant when $\alpha \sim \mathrm{Sg}(a, 1, m)$ and $a\in \mathds{N}$ and $m \geq 3$ is 
$$
\mathcal{S}_{a, 1, m} = \sum_{j=1}^{m-1} (-1)^{a+j}\frac{j^{a-2}\log{j}}{\Gamma(j)\Gamma(m-j)}.
$$
\end{corollary}

\subsection{Proofs of Theorem~\ref{theo:NormConst}}
We break down the proof of Theorem~\ref{theo:NormConst} into three steps to ease readability. First, we prove that the quantity $\alpha^{a-1}/\{(\alpha)_m\}^{b}$ can be rewritten as a sum of partial fractions with Lemma~\ref{lem:PartialFraction}. Second, we illustrate how this decomposition is useful to evaluate the normalizing constant integral via Lemma~\ref{lem:PartialFractionIntegral}. Then, the proof of the statement follows using Faà di Bruno's formula.
\begin{lemma}\label{lem:PartialFraction}
Let $a$ and $b$ be integers. Then, we can write
\begin{equation}\label{eq:PartialFractions}
    \frac{\alpha^{a-1}}{\{(\alpha)_m\}^b} =
    \sum_{j=1}^{m-1}\sum_{s=1}^{b} \frac{A_{s,j}}{(\alpha+j)^s},
\end{equation}
where $A_{s,j} = \rho_j^{(b-s)}(-j)/(b-s)!$ and $\rho^{(d)}_j(\alpha)$ is the $d^{\mathrm{th}}$ derivative of the function $\rho_j(\alpha) = \alpha^{a-b-1}/\{(\alpha+1)_{j-1}(\alpha + j + 1)_{m-j-1}\}^b$.
\end{lemma}
\begin{proof}
From the definition of ascending factorial, we can write
$$\frac{\alpha^{a-1}}{\{(\alpha)_m\}^{b}} = \frac{\alpha^{a-b-1}}{\prod_{i=1}^{m-1}(\alpha + i)^b},$$
which is ratio of polynomials whose roots for the denominators are $-1, \ldots, -m + 1$. Following the algorithm of Section 2.102, page 66 in \citetSupp{gradshteyn2007}, we can rewrite the above as the sum of partial fractions in equation~\eqref{eq:PartialFractions}, where the coefficients $A_{s,j}$ of the expansion depend on the derivatives of the function
$$
\rho_j(\alpha) = \frac{\alpha^{a-b-1}}{\prod_{i=1}^{m-1}(\alpha + i)^b}(\alpha + j)^b, \qquad (j = 1, \ldots, m-1),
$$
evaluated at the solution $\alpha =-j$. In particular, we have $A_{s, j} = \rho_j^{(b-s)}(-j)/(b-s)!$ We calculate their exact values below, after noticing that
\begin{equation}\label{eq:rhoj}
\rho_j(\alpha)= \frac{\alpha^{a-b-1}}{\prod_{i=1}^{j-1}(\alpha + i)^b\prod_{i=j+1}^{m-1}(\alpha + i)^b} = \frac{\alpha^{a-b-1}}{\{(\alpha + 1)_{j-1}(\alpha+j+1)_{m-j-1}\}^b}.
\end{equation}
\end{proof}

\begin{lemma}\label{lem:PartialFractionIntegral}
The normalizing constant of the Stirling-gamma $\alpha \sim \mathrm{Sg}(a, b, m)$ where $a, b \in \mathds{N}$ can be expressed as
\begin{equation*}
    \mathcal{S}_{a, b, m} = \sum_{j=1}^{m-1}\sum_{s=1}^{b} A_{s,j} \phi_s(j), \qquad \phi_s(j) =
\begin{cases}
-\log{j}, &s = 1,\\
j^{1-s}/(s-1) &s=2, 3, \ldots \\
\end{cases}    
\end{equation*}

where $A_{s,j}$ are defined in Lemma~\ref{lem:PartialFraction}. 
\end{lemma}
\begin{proof} 
Recall that $\int_{\mathds{R}_+}1/(\alpha + j)^s\text{d}\alpha = 1/\{(s-1)j^{s-1}\}$ for $s > 1$, while $$\int_{\mathds{R}_+}\frac{1}{\alpha + j}\text{d}\alpha = \lim_{\alpha \to\infty}\log(\alpha + j) - \log{j}.$$ Define the functions
$$
\phi_1(j) = -\log{j}, \qquad \phi_s(j) = \frac{1}{(s-1)j^{s-1}}, \quad s = 2,3, \ldots
$$
From Lemma~\ref{lem:PartialFraction}, we have
\begin{align*}
\mathcal{S}_{a, b, m} = \int_{\mathds{R}_+} \frac{\alpha^{a-1}}{\{(\alpha)_{m}\}^b}\mathrm{d}\alpha = 
\int_{\mathds{R}_+} 
    \sum_{j=1}^{m-1}\sum_{s=1}^{b} \frac{A_{s,j}}{(\alpha+j)^s} \text{d}\alpha
    = \sum_{j=1}^{m-1}\sum_{s=1}^{b} A_{s,j} \phi_s(j)
\end{align*}

The last equality holds because $\sum_{j = 1}^{m-1} A_{1, j} = 0$, which makes the limit of the sum of logarithms necessarily equal to zero. This can be shown by contradiction. If $\sum_{j = 1}^{m-1} A_{1, j} \neq 0$, then necessarily $\left\vert \sum_{j = 1}^{m-1} A_{1, j} \lim_{\alpha \to\infty} \log{(\alpha + j)} \right\vert = \infty$, which implies that $|\mathcal{S}_{a,b,m}| = \infty$. However, this contradicts Proposition~\ref{pro:NormConst_finite}, which states that $0 < \mathcal{S}_{a,b,m} < \infty$ for appropriate choices of $a$, $b$ and $m$. Hence, the divergence of each logarithmic term is compensated by the alternating sum. For an alternative proof of why this happens, refer to \citetSupp{Zhu_Luo_2021} and references therein. 
\end{proof}

\begin{proof}[Proof of Theorem~\ref{theo:NormConst}]
Lemma~\ref{lem:PartialFraction} and Lemma~\ref{lem:PartialFractionIntegral} show that we can write the normalizing constant as a sum of logarithms. It remains to calculate the values for the coefficients $A_{s, j} = \rho_j^{(b-s)}(-j)/(b-s)!$ We start by rewriting equation~\eqref{eq:rhoj} as
$$
\rho_j(\alpha) = \exp\left[(a-b-1)\log{\alpha} - b\left\{\log(\alpha+1)_{j-1} + \log(\alpha+j+1)_{m-j-1} \right\} \right].
$$
Recalling that $\frac{\text{d}}{\text{d}x}\log(x)_n = \psi(x + n) -\psi(x)$ and that $\psi(x+1) = \psi(x) + 1/x$, we have that 
\begin{align*}
\rho_j'(\alpha) = \rho_j(\alpha)h_j(\alpha), \quad h_j(\alpha) = \frac{a-1}{\alpha} - b\{\psi(\alpha + m) -\psi(\alpha) - \psi(\alpha + j + 1) + \psi(\alpha + j)\}.
\end{align*}
Then, the $s^{\textrm{th}}$ derivative of $\rho_j(\alpha)$ can be expressed via Faà di Bruno's formula as
\begin{equation}\label{eq:Derivatives}
\rho_j^{(s)}(\alpha) = \rho_j(\alpha)B_s\{h_j(\alpha), h_j'(\alpha), \ldots, h_j^{(s-1)}(\alpha)\},
\end{equation}
where 
\begin{equation}\label{eq:hdj}
h_j^{(d)}(\alpha) = (-1)^{d}d!\frac{a-1}{\alpha^{d+1}} - b\{\psi^{(d)}(\alpha + m) -\psi^{(d)}(\alpha) - \psi^{(d)}(\alpha + j + 1) + \psi^{(d)}(\alpha + j)\} 
\end{equation}
for $d = 0, \ldots, s-1$, and
\begin{equation}\label{eq:Bell}
    B_s(x_1, \ldots, x_s) = \sum_{(j_1, \ldots, j_s) \in I_s} \frac{s!}{j_1!j_2!\cdots j_s!} \left(\frac{x_1}{1!}\right)^{j_1}\left(\frac{x_2}{2!}\right)^{j_2}\cdots \left(\frac{x_s}{s!}\right)^{j_s}
\end{equation}
is the complete exponential Bell polynomial of order $s$ and $I_s$ is the set of all nonnegative integers $(j_1, \ldots, j_s)$ that satisfy $j_1 + 2j_2+\ldots+ sj_s = s$. It remains to evaluate the function $\rho_j^{(s)}(\alpha)$ when $\alpha = -j$. We start by noticing that
\begin{align*}
\prod_{i=1}^{j-1}(-j + i)^b &= \{(1-j)(2-j)\cdots (-2)(-1)\}^b = (-1)^{bj}\Gamma(j)^b,\\
\prod_{i=j+1}^{m-1}(-j + i)^b &= \{1 \cdot 2 \cdot 3 \cdots (m-j -1)\}^b = \Gamma(m-j)^b.
\end{align*}
Plugging in the above in equation~\eqref{eq:rhoj}, we have
\begin{align}\label{eq:rho_minusj}
\rho_j(-j) = \frac{(-j)^{a-b-1}}{\prod_{i=1}^{j-1}(-j + i)^b\prod_{i=j+1}^{m-1}(-j + i)^b} = \frac{(-1)^{a-b(j+1) - 1} j ^{a-b-1}}{\Gamma(j)^b \Gamma(m-j)^b}.
\end{align}
Moreover, we have that the derivatives in equation~\eqref{eq:hdj} can be rewritten as
\begin{equation}
h_j^{(d)}(\alpha) = (-1)^{d}d!\left[\frac{a-1}{\alpha^{d+1}} - b\left\{\sum_{i = 0}^{j-1} \frac{1}{(\alpha + i)^{d+1}} + \sum_{i = j+1}^{m-1} \frac{1}{(\alpha + i)^{d+1}}\right\}\right].
\end{equation}
Calling $h_{j, d} = h_j^{(d)}(-j)$, we then have that 
\begin{equation}\label{eq:h}
\frac{h_{j, d+1}}{d!} = -\frac{(a-1)}{j^{d+1}} - b(H_{m-j-1, d+1} - H_{j, d+1}), 
\end{equation}
with $H_{j, s} = \sum_{i=1}^j 1/i^s$ the $j^{\text{th}}$ generalized harmonic number of order $s$. Plugging equations~\eqref{eq:rho_minusj} and \eqref{eq:h} into \eqref{eq:Derivatives} and recalling that $A_{s, j} = \rho_j^{(b-s)}(-j)/(b-s)!$, we write the partial fraction decomposition coefficients as 
$$
A_{s, j} = \frac{1}{(b-s)!}\frac{(-1)^{a-b(j+1) - 1} j ^{a-b-1}}{\{\Gamma(j) \Gamma(m-j)\}^b} B_{b-s}(h_{j, 1}, h_{j,2}, \ldots, h_{j,b-s}).
$$
Combining this expression with Lemma~\ref{lem:PartialFractionIntegral} yields the final form
\begin{align*}
\mathcal{S}_{a, b, m} &= \sum_{j=1}^{m-1}\sum_{s=1}^{b} A_{s,j} \phi_s(j) \\
&= \sum_{j=1}^{m-1}\sum_{s=1}^{b}\frac{(-1)^{a-b(j+1) - 1} j ^{a-b-1}}{\{\Gamma(j) \Gamma(m-j)\}^b} \frac{B_{b-s}(h_{j, 1}, h_{j,2}, \ldots, h_{j,b-s})}{{(b-s)!}}\phi_s(j).
\end{align*}
Rearranging the terms and collecting the Bell polynomials into the quantity
$$
\mathscr{S}_{b, j}(x_1, \ldots, x_b) =  \sum_{s=1}^{b} \frac{B_{b-s}(x_{1}, \ldots, x_{b-s})}{(b-s)!}\phi_s(j),
$$
yields the desired result.
\end{proof}
\subsection{Proof of Corollary~\ref{cor:NormConst_b1}}
\begin{proof}
Setting $b = 1$ in the above statement proves the statement. To see why, recall that $A_{s, j} = \rho_{j}^{(b-s)}(-j)/(b-s)!$ Since $b = 1$, we only have $A_{1, j} =  \rho_{j}(-j)$, whose general formula is provided in equation~\eqref{eq:rho_minusj}.
\end{proof}

\subsection{Proof of Theorem 2}
The proof of Theorem 2 follows the same reasoning as the one of Theorem A1. We discuss the highlights with the help of the following statements. 
\begin{lemma}\label{lem:PartialFractionPost}
Let $a, b, m, n, k \in \mathds{N}$ with $1<a/b<m$, and $1 \geq k \geq n$, and call $M = \min\{n, m\}$ and $\ell = |n - m|$. Then, we can write
\begin{equation}\label{eq:PartialFractionPost}
    \frac{\alpha^{a + k -1}}{\{(\alpha)_m\}^b(\alpha)_n} =
    \sum_{j=1}^{M-1}\sum_{s=1}^{b+1} \frac{T_{s,j}}{(\alpha+j)^s} + \sum_{i = 0}^{\ell - 1} \frac{U_{i}}{\alpha + M +i},
\end{equation}
where $T_{s,j} = \tau_j^{(b + 1-s)}(-j)/(b + 1-s)!$ and $\tau^{(d)}_j(\alpha)$ is the $d^{\mathrm{th}}$ derivative of the function 
\begin{equation}\label{eq:Tsj}
\tau_j(\alpha) = \frac{\alpha^{a+k-b-2}}{\{(\alpha+1)_{j-1}(\alpha + j + 1)_{M-j-1}\}^{b + 1}(\alpha+M)_\ell},
\end{equation}
$U_i = u_i(-M-i)$ and 
\begin{equation}\label{eq:ui}
u_i(\alpha) = \frac{\alpha^{a+k-b-2}}{\prod_{j = 1}^{M-1}(\alpha + j)^{b + 1}\prod_{t = 0}^{i - 1}(\alpha + M + t)\prod_{v = i + 1}^{\ell - 1}(\alpha + M + v)}
\end{equation}
\end{lemma}

\begin{proof}
The proof follows from the same reasoning discussed in Lemma~\ref{lem:PartialFraction}, which is a direct consequence of the algorithm of Section 2.102, page 66 in \citetSupp{gradshteyn2007}, after writing that
\begin{equation}\label{eq:fraction_rewrite}
\frac{\alpha^{a + k -1}}{\{(\alpha)_m\}^b(\alpha)_n} = \frac{\alpha^{a + k -1}}{\{(\alpha)_M\}^{b+1}(\alpha + M)_\ell} = \frac{\alpha^{a - b + k - 2}}{\prod_{j = 1}^{M-1}(\alpha + j)^{b + 1} \prod_{i = 0}^{\ell - 1}(\alpha + M + i)}.
\end{equation}
The quantities $\tau_j(\alpha)$ and $u_j(\alpha)$ are obtained by multiplying equation~\eqref{eq:fraction_rewrite} by $(\alpha + j)^{b+1}$ and $(\alpha + M + i)$, respectively, and simplifying appropriately.
\end{proof}

\begin{lemma}\label{lem:PartialFractionIntegral_post}
The coefficients $\mathscr{V}_{a, b, m}(n, k)$  when $a, b \in \mathds{N}$ are expressed as follows:
\begin{equation*}
    \mathscr{V}_{a, b, m}(n, k) = \sum_{j=1}^{M-1}\sum_{s=1}^{b+1} T_{s,j}\phi_s{(j)} - \sum_{i = 0}^{\ell - 1} U_{i}\log{(M + i)},
\end{equation*}
where $T_{s,j}$ and $U_i$ are defined in Lemma~\ref{lem:PartialFractionPost} and the coefficients $\phi_s(j)$ are defined in Lemma~\ref{lem:PartialFractionIntegral}. 
\end{lemma}
\begin{proof}
Follow the same line of reasoning as the proof of Lemma~\ref{lem:PartialFractionIntegral}, since we have rational functions as integrands. In particular, we point out that the coefficients multiplying the logarithms are such that $\sum_{j = 1}^{M-1} T_{1, j} + \sum_{i = 0}^{\ell - 1} U_i = 0$. Again, this must happen because the diverging logarithms resulting from the integration must cancel each other out because $\mathscr{V}_{a, b, m}(n, k) < \infty$ by definition. 
\end{proof}

We are now ready to provide proof of the final statement. 
\begin{proof}[Proof of Theorem A2]
First of all, we provide a simpler expression for $U_i$. This is equal to $u_i(-M-i)$ in equation~\eqref{eq:ui}. In particular, the quantities at the denominator simplify as
\begin{align*}
\prod_{j = 1}^{M-1}(-M-i+j)^{b +1} &= \{(\!-M\!-i\!+1)(\!-M\!-i\!+2)\cdots(\!-i\!-1)\}^{b+1} = \{(i+1)_{M-1}\}^{b + 1} \\
\prod_{t = 0}^{i-1}(-M-i + t) &= (-i)(-i + 1)\cdots(-2)(-1) = 
(-1)^i\Gamma(i+1)\\
\prod_{v = i+1}^{\ell-1}(-M-i+ v) &= (1)(2)\cdots(\ell -1 + i) = \Gamma(\ell - i).
\end{align*}
This implies that 
$$
U_i = (-1)^{a-b+k-2+i}\frac{(M+i)^{a-b+k-2}}{\{(i+1)_{M-1}\}^{b + 1}\Gamma(i+1)\Gamma(\ell - i)},
$$
which is the generic coefficient in the second sum of logarithms. As for $T_{s,j}$, we rely on a similar argument as the proof of Theorem A1, which depends on Faà di Bruno's formula. Rewriting $\tau_j(\alpha)$ as 
\begin{align*}
\tau_j(\alpha) &= \exp\{\log{\tau_j(\alpha)}\} \\
&= \exp[(a+k-b-2)\log{\alpha} - (b+1)\{\log(\alpha + 1)_{j-1} + \log{(\alpha + j + 1)_{M-j-1}}\} \\
&\qquad - \log{(\alpha + M)_\ell}],
\end{align*}
we have that $\tau'_{j}(\alpha) = \tau_j(\alpha)g_j(\alpha)$ with 
\begin{align*}
g_j(\alpha) &= \frac{a+k -1}{\alpha} - (b+1)\{\psi(\alpha+j) - \psi(\alpha) + \psi(\alpha + M) - \psi(\alpha + j +1)\} \\
&\qquad- \psi(\alpha + M + \ell) + \psi(\alpha + M)
\end{align*}
whose $d$th derivative is equal to
\begin{align*}
g^{(d)}_j(\alpha) & = (-1)^dd!\frac{a+k -1}{\alpha^{d+1}} - (b+1)\{\psi^{(d)}(\alpha+j) - \psi^{(d)}(\alpha) + \psi^{(d)}(\alpha + M) \\
&\qquad - \psi^{(d)}(\alpha + j +1)\} - \psi^{(d)}(\alpha + M + \ell) + \psi^{(d)}(\alpha + M) \\
&=  (-1)^dd!\!\left[\frac{a\!+\!k\!-\!1}{\alpha^{d+1}}\! - \!(b\!+\!1)\!\left\{\sum_{i = 0}^{j-1}\frac{1}{(\alpha\! +\! i)^{d+1}} \!+\!\sum_{i = j+1}^{M-1}\frac{1}{(\alpha\! +\! i)^{d+1}}\right\}\! +\! \sum_{i = 0}^{\ell-1}\frac{1}{\alpha\!+ \!M\! +\! i}\right].
\end{align*}
Using Faà di Bruno's formula and recalling the complete exponential Bell polynomial in equation~\eqref{eq:Bell}, we obtain that 
\begin{equation}\label{eq:tj_bell}
\tau_{j}^{(d)}(\alpha) = \tau_{j}(\alpha)B_d\left(g_j(\alpha),g'_j(\alpha), \ldots, g^{(d-1)}_j(\alpha)\right).
\end{equation}
We are finally ready to calculate $T_{s, j} = \tau_j^{(b+1-s)}(-j)/(b+1-s)!$. Calling $g_{j, d+1} = g_j^{(d)}(-j)$, we can write that
\begin{equation}\label{eq:gjd}
\frac{g_{j, d+1}}{d!} = -\frac{a+k-1}{j^{d+1}} + b H_{M-j-1, d+1} - (b+1)H_{j, d+1} + H_{M-j+\ell-1, d+1},
\end{equation}
since $\sum_{i=0}^{\ell-1}1/(M+i-j)^{d+1} =  H_{M-j+\ell-1, d+1} -H_{M-j-1, d+1}$, and $H_{j,s} = \sum_{i = 1}^{j}1/i^s$ is again the $j$th generalized harmonic number of order $s$. With similar calculations as the one for equation~\eqref{eq:rho_minusj}, we also have that
\begin{equation}\label{eq:tau_minusj}
\begin{split}
    \tau_j(-j) &= \frac{(-j)^{a-b + k-2}}{\{\prod_{i=1}^{j-1}(-j + i)\prod_{i=j+1}^{m-1}(-j + i)\}^{b+1}(M-j)_\ell} \\
    &= (-1)^{a-b + k -2 + b(j+1)} \frac{j ^{a-b + k -2}}{\{\Gamma(j) \Gamma(m-j)\}^{b+1}(M-j)_\ell}.
\end{split} 
\end{equation}
Plugging equations~\eqref{eq:tau_minusj}, \eqref{eq:gjd} and \eqref{eq:tj_bell} into the formula for $T_{s, j}$ yiels
$$
T_{s, j} = \frac{(-1)^{a-b + k -2 + b(j+1)} j ^{a-b + k -2}}{\{\Gamma(j) \Gamma(m-j)\}^{b+1}(M-j)_\ell} \frac{B_{b+1-s}(g_{j, 1}, g_{j, 2}, \ldots, g_{j,b+1-s})}{(b+1-s)!}.
$$
The rest of the proof follows by regrouping the coefficients in a similar manner as the one in the proof for Theorem A1. 
\end{proof}

\section{Further clustering implications}\label{sec:FurtherResults}

\subsection{Relationship with clustering consistency results of \citet{Ascolani_2022}}
In this Section, we discuss the implications of adopting the Stirling-gamma prior for $\alpha$ in terms of the consistency of the DPM mixture model. Within this setting, \citetSupp{Ascolani_2022} proposes a class of priors that, under certain strict requirements on the data and on the mixture kernels, lead to the desired consistency result. In particular, they impose the following assumptions on $\pi(\alpha)$:
\begin{enumerate}[label=A\arabic*]
\item continuous density with respect to the Lebesgue measure;
\item polynomial behavior around the origin: There exist $\epsilon, \delta, \beta$ such that, for all $\alpha \in (0, \epsilon)$, it holds that $\alpha^\beta/\delta \leq \pi(\alpha) \leq \delta \alpha^\beta$;
\item subfactorial moments: there exist $D, \nu>0$ such that  $\mathds{E}(\alpha^s) < D\rho^{-s}\Gamma(\nu + s +1)$ for every $s\geq1$.
\end{enumerate}
The Stirling-gamma distribution $\alpha \sim \mathrm{Sg}(a, b, m)$ with a fixed reference sample size $m$ satisfies assumptions A1 and A2, but not A3. In particular, A1 trivially holds since the Stirling-gamma is a continuous random variable. To see that A2 is met, pick $\epsilon = 1$ and recall that $(\alpha)_m = \alpha \prod_{j = 1}^{m-1}(\alpha + j)$. Then, for $\alpha \in (0, 1)$, we have, 
$$
\frac{1}{\mathcal{S}_{a, b, m}}\frac{\alpha^{a-b-1}}{\Gamma(m+1)^b} \leq \frac{1}{\mathcal{S}_{a, b, m}}\frac{\alpha^{a-b-1}}{\prod_{j = 1}^{m-1} (\alpha + j)^b} \leq \frac{1}{\mathcal{S}_{a, b, m}}\frac{\alpha^{a-b-1}}{\Gamma(m)^b},
$$
where the quantity in the center is the density of the Stirling-gamma. We can then set $\beta = a-b-1$, and pick any $\delta > \mathrm{max}\{\mathcal{S}_{a, b, m}\Gamma(m+1)^b, 1/[\mathcal{S}_{a, b, m}\Gamma(m)^b]\}$ to meet the requirements of A2. Finally, recall that from Proposition~\ref{pro:moments} in the main paper, we have that $\mathds{E}(\alpha ^s) = \infty$ whenever $s > ma-b$. This means that A3 does not hold for every 
$s\geq 1$, preventing us from directly applying the results in Theorems 1 to 3 in \citetSupp{Ascolani_2022}

\subsection{Applications of the Stirling-gamma to other partitions models}
We now discuss two additional contexts from the Bayesian nonparametric literature where the Stirling-gamma can be potentially useful as a prior for some hyperparameters.

A first simple extension dwells in the \emph{extended stochastic block model} framework of \citetSupp{Legramanti2022}. In particular, they leverage upon the covariate-dependent product partitions models in  \citetSupp{Park_Dunson_2010} and \citetSupp{Muller_quintana_rosner_2011} and consider an extended version of the exchangeable partition probability function of the Dirichlet process in equation~\eqref{eq:dp_eppf}, namely
\begin{equation}\label{eq:esbm_DP}
\mathds{P}(\Pi_n = \{C_1,\ldots, C_k\}\mid \mathbf{X}, \alpha) = \frac{\alpha^k}{(\alpha)_n} \prod_{j = 1}^k g(\mathbf{X}_k)(n_j - 1)!,
\end{equation}
where $g(X)$ is a \textit{cohesion function} that depends on cluster-specific attributes $\mathbf{X}_k$. It is straightforward to see that equation~\eqref{eq:esbm_DP} depends on $\alpha$ only through the size of the partition $k$. Then, from trivial calculations based on Proposition~\ref{pro:conjugacy}, imposing $\alpha\sim\mathrm{Sg}(a,b,n)$ as a prior over the precision still yields $(\alpha\mid \Pi_n, \mathbf{X}) \sim \mathrm{Sg}(a+k,b+1,n)$. This implies that the Stirling-gamma can be a suitable prior for preventing over-clustering in Dirichlet process-based extended stochastic block models; see \citetSupp{Legramanti2022}.

A second useful example lies in the \emph{generalized mixtures of finite mixtures} introduced by \citetSupp{FruFru_2021}. These are mixture models where the prior over the mixing weights depends on the total number of components as well. The implied exchangeable partition function, which is not of the Gibbs-type class, is presented in their Theorem 2.2, and is equal to
\begin{equation*}
\mathds{P}(\Pi_n = \{C_1,\ldots, C_{K_+}\}\mid \alpha) = \frac{\alpha^{K_+}}{(\alpha)_n} \prod_{j=1}^{K_+}\Gamma(n_j) \times\sum_{k = K_+}^{\infty} \pi(k) \prod_{j=1}^{K_+}\frac{\Gamma(n_j + \alpha/k)(k-j+1)}{\Gamma(1+\alpha/k) \Gamma(n_j) k}, \end{equation*}
where $K_+$ indicates the total clusters, while $\pi(k)$ is a prior over the number of mixture components. It is easy to see that when $\pi(k) = \delta_\infty$, then the model above reduces to the partition of the Dirichlet process in equation~\eqref{eq:dp_eppf}. However, since $\alpha$ enters the infinite sum, the Stirling-gamma is only \emph{quasi-conjugate}.  Nevertheless, one can still perform inference under their framework via Metropolis-Hastings moves. Interestingly, \citetSupp{FruFru_2021} specify an F-distribution, which is again a heavy-tailed prior, for $\alpha$, indicating this as a better choice than the gamma of \citetSupp{Escobar_West_1995}. Hence, we can seamlessly use a prior $\alpha\sim \mathrm{Sg}(a, b, m)$ with small $m$ in place of the proposed F-distribution and still retain similar interpretability as in the Dirichlet process case when $K$ is sufficiently large while maintaining the desirable heavy-tailedness mentioned in their Section 4.3.

\section{Simulations}\label{sec:simul}

\subsection{Details of the simulation in Figure~\ref{fig:figure1}}
The data in the left panel of Figure~\ref{fig:figure1} consists of $n = 800$ observations generated independently from a mixture of four equally weighted bivariate normal distributions in the same way described in Section~\ref{subsec:Gamma_vs_Sg}. Four different scenarios are considered with respect to $\alpha$: ``Fixed, low'' sets $\alpha = 1$, ``Random, high'' sets $\alpha = 5$, ``Random, low'' lets $\alpha \sim \mathrm{Sg}(0.73, 0.1)$ and ``Random, high'' lets $\alpha \sim \mathrm{Sg}(2.6, 0.1)$. The induced distribution on $K_n$ has mean $\mathds{E}(K_n) = 7.26$ in low cases and $\mathds{E}(K_n) = 26$ in high ones. Inference on the number of clusters $K_n$ in each scenario is performed by running a marginal Gibbs sampler as in Algorithm 3 in \citetSupp{Neal_2000} for 20,000 iterations, discarding the first 5,000 as burn-in. 

\subsection{Pooling $\alpha$ across repeated networks}
\begin{figure}[tb]
\centering
\includegraphics[width=\linewidth]{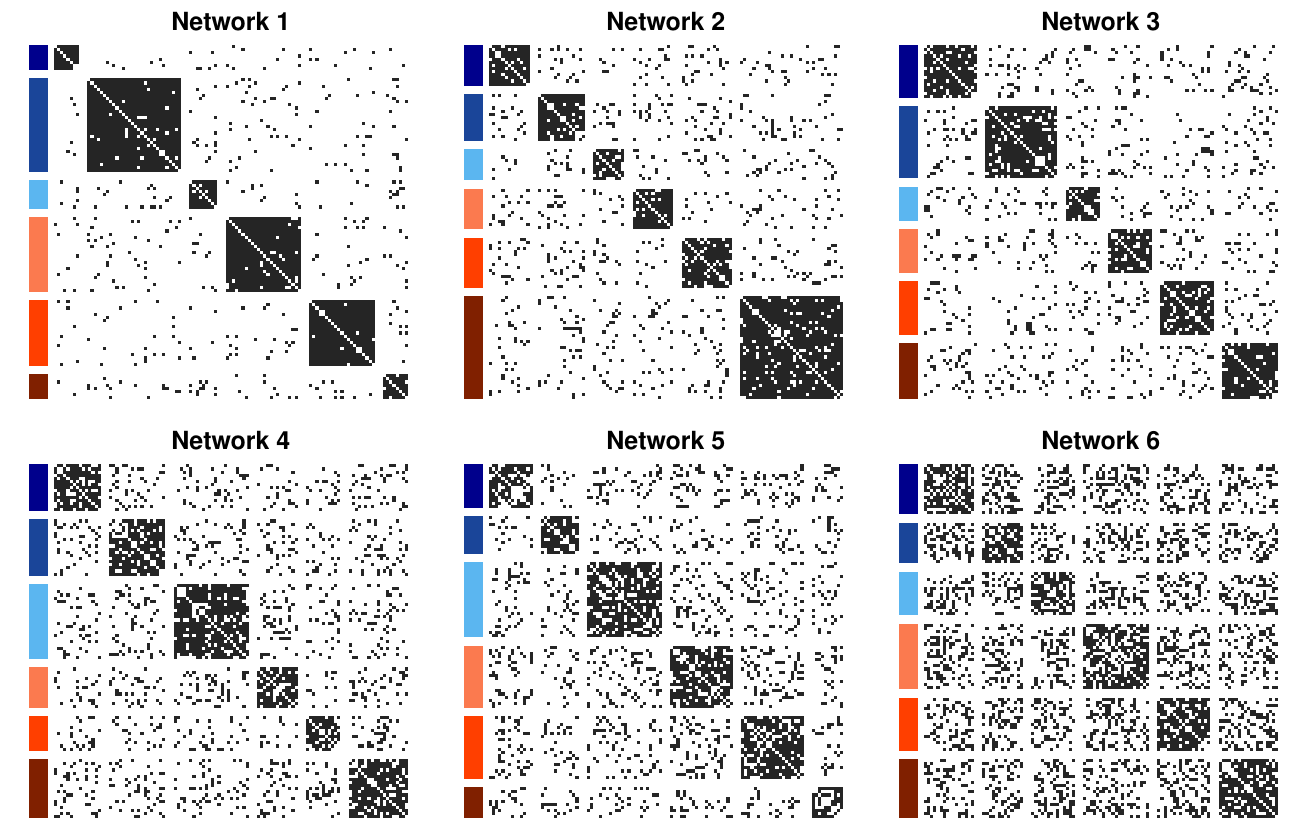}
\caption{Simulated networks of size $n=100$ nodes. Columns and rows represent nodes of each network, and black dots indicate the edges. Nodes are sorted according to the true cluster assignment, highlighted by the different colors on the left of each plot.}
\label{fig:Networks} 
\end{figure}

In this Section, we present a simulation study within the same \emph{population of partition} framework introduced in Section~\ref{sec:inference} of the main paper. In particular, our goal is to show how the conjugate Stirling-gamma prior can lead to borrowing of information when inferring the latent partition across multiple networks, thus reducing uncertainty. Consider the same stochastic block model setting of Section 4, namely
\begin{equation}\label{eq:sbm_suppl}
\mathds{P}(X_{i, j,s} = 1 \mid Z_{i, s} = h, Z_{j,s} = h', \nu) = \nu_{h, h',s}, \quad \nu_{h, h',s} \sim \mathrm{Be}(1, 1),
\end{equation}
where $X_{i, j,s}$  is a binary random variable
indicating an edge between nodes $i$ and $j$
in networks $s = 1, \ldots N$. The variables $Z_{i, s}$ denote cluster assignment, with $Z_{i, s} = h$ if and only if $i \in C_{h, s}$ in network $s$, and $\nu_{h, h',s}$ denotes the edge probabilities in the block identified by clusters $C_{h, s}$ and $C_{h',s}$. We model the latent partition in each network independently as follows:
$$
\mathds{P}(\Pi_{n, s} = \{C_{1, s}, \ldots, C_{k_s, s}\} \mid \alpha_s) = \frac{\alpha_s^{k_s}}{(\alpha_s)_n}\prod_{j = 1}^{k_s} (n_{j, s} - 1)!  \quad (s = 1, \ldots, N),
$$
where $\alpha_s$ is the precision parameter specific to partition $\alpha_s$.

Within this framework, we are interested in investigating the impact of different choices of precision parameters $\alpha_1,\ldots, \alpha_N$ on the inferred latent partition. We consider three priors: $\alpha_s$ is fixed and equal across networks, $\alpha_s$ is random with $\alpha_s \sim \mathrm{Sg}(a, b, n)$ separately for each network, and the precision is pooled across networks, namely $\alpha_1 = \ldots = \alpha_N = \alpha \sim \mathrm{Sg}(a, b, n)$. In the third case, the shared $\alpha$ induces borrowing of information since the number of clusters $k_s$ in every network contributes to the posterior distribution in Theorem 4. 

We simulate $N = 6$ networks of $n=100$ nodes from the stochastic block model in equation~\eqref{eq:sbm_suppl}. The true partition is generated by randomly dividing the nodes between six clusters with assignment probabilities drawn from a Dirichlet distribution $\mathrm{Dir}(10,10,10,10,10,10)$. Binary edges are independently simulated with probabilities $(\nu_{h, h,1}, \ldots, \nu_{h, h,N})= (0.95, 0.90, 0.85, 0.80, 0.75, 0.70)$ for nodes within the same cluster, and $(\nu_{h, h',1}, \ldots, \nu_{h, h,N})= (0.05, 0.10, 0.10, 0.15, 0.15, 0.30)$ for any $h \neq h'$. This allows each network to have a different block structure with decreasing signal-to-noise ratios. As such, we expect to infer the true communities in Networks 5 and 6 with a higher uncertainty than in Networks 1 and 2. Figure~\ref{fig:Networks} displays the six generated datasets. Black points indicate an edge between each pair of nodes. Rows and columns have been sorted according to the true cluster assignment for better visualization.
We set $\alpha_s = 7.5$ in the fixed case and $a = 6$ and $b = 0.3$ in random and pooled cases, so that $\mathds{E}(K_n) = 20$ in all priors. Inference is performed by running a collapsed Gibbs sampler as in \citetSupp{Legramanti2022} for 10,000 iterations, treating the first 2,000 as burn-in. The full conditional for $\alpha_s$ in the random case and for $\alpha$ in the pooled case are reported in Proposition~4 and Theorem~4 in the main manuscript, respectively.

Figure~\ref{fig:SimKn} displays the posterior distribution of the number of detected clusters $K_n$ in each dataset for the three choices of precision parameter. Except for Network 6, the posterior mode of $K_n$ coincides with the truth in each model. However, the pooled case shows lower uncertainty than the random one, thanks to the borrowing of information granted by the common $\alpha$. In the fixed cases, instead, $K_n$ explodes as the signal-to-noise ratio decreases. This is particularly evident in Network 6, which confirms the lack of robustness of Dirichlet process mixtures with fixed $\alpha$. To further highlight these differences, we calculate the average adjusted Rand index for the posterior partition retrieved by the three models with respect to the truth. This equals $0.943$ for the pooled case, $0.940$ for the random, and $0.929$ for the fixed, indicating that pooling $\alpha$ yields a better estimate across networks. 

\begin{figure}[tb]
\centering
\includegraphics[width = \linewidth]{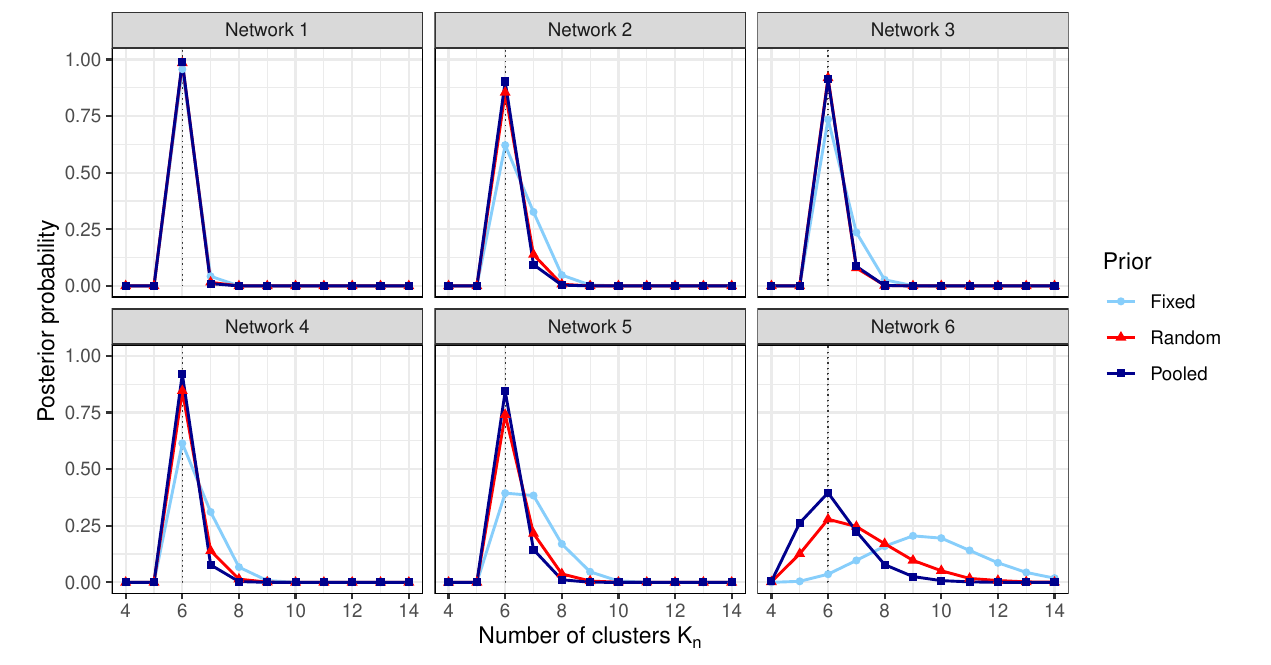}
\caption{Posterior distribution of the number of clusters $K_n$ detected in each simulated network in the three cases: $\alpha_s = 7.5$ (light blue), $\alpha_s \sim \mathrm{Sg}(6, 0.3, 100)$ independently in each network  (red), and $\alpha_s = \alpha\sim \mathrm{Sg}(6, 0.3, 100)$ (blue). The dotted vertical line highlights the true number of communities.}
\label{fig:SimKn}
\end{figure}

\subsection{Effective sample sizes in the standard normal Dirichlet process mixture simulation}
We now briefly report on the effective sample sizes for the simulation in Section~\ref{subsec:normalDPM}. These are illustrated in Figure~\ref{fig:ESS}, where each boxplot displays the results of the 40 replicates for each prior and each $n$. Values are calculated using the 2000 post-burn-in samples drawn using Algorithm 3 of \citetSupp{Neal_2000}. We notice the following. First, values appear similar across $n$ in all three priors, indicating that the sample size does not influence sampling efficiency in this setting. Second, the weakly informative prior  $\mathrm{Sg}(1, 0.25, n)$ shows consistently higher effective sample sizes than the two alternatives, while the more informative $\mathrm{Sg}(4, 1, n)$ attains the worst performance. This is at odds with the results shown in Table~\ref{tab:effectiveSize}  for the simulation in Section~\ref{subsec:Gamma_vs_Sg}, where this prior achieves the best performance. However, these differences are likely due to model- and setting-specific variabilities. Finally, we detect again minor computational differences between the gamma and the Stirling-gamma.

\begin{figure}
    \centering
    \includegraphics[width = \linewidth]{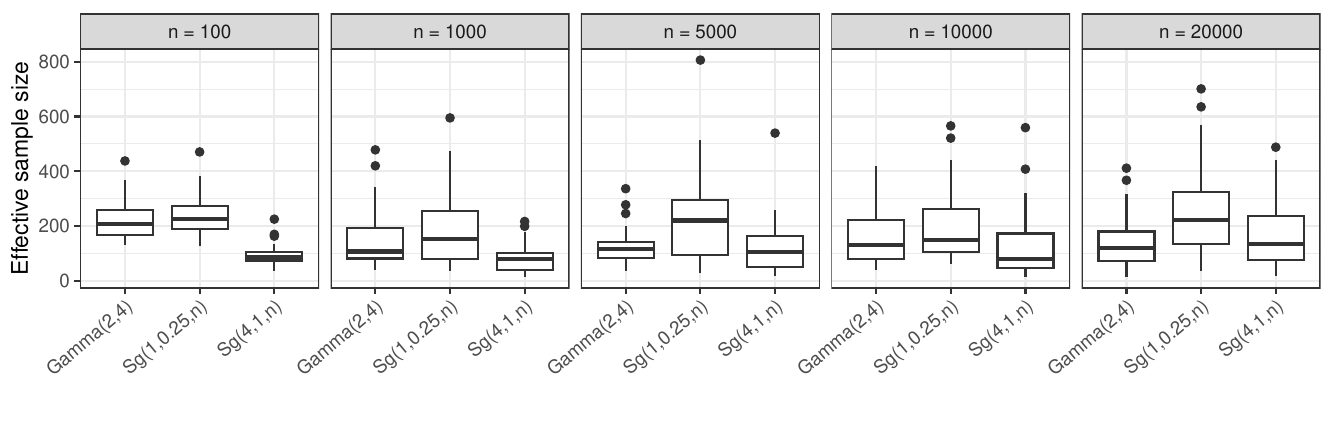}
    \caption{Effective sample sizes for the posterior of $\alpha$ in the standard normal simulation in Section~\ref{subsec:normalDPM}, across values of $n$ and choices of priors. Values are calculated using the 2000 post-burn-in samples in each of the 40 replicates per scenario.}
    \label{fig:ESS}
\end{figure}

\bibliographystyleSupp{chicago} 
\bibliographySupp{references} 

\end{document}